\documentclass[12pt,letterpaper]{article}
\usepackage{lipsum}
\usepackage{authblk}

\usepackage[normalem]{ulem}

\usepackage{amsthm}
\usepackage{newpxtext,newpxmath}
\usepackage[margin=1.25in]{geometry}

\usepackage{subcaption}

\usepackage[dvipsnames]{xcolor}

\usepackage{algorithm}
\usepackage{algpseudocode}
\usepackage{color}
\usepackage{graphicx}
\usepackage{comment}
\usepackage{mathtools}
\usepackage{enumitem}
\usepackage{url}
\usepackage{float}

\usepackage{fancyref}
\usepackage{hyperref}

\usepackage{cleveref}

\usepackage{quantikz}

\newtheorem{theorem}{Theorem}[section]
\newtheorem{lemma}[theorem]{Lemma}

\newtheorem{corollary}[theorem]{Corollary}

\newcommand{\lrb}[1]{\left ( #1 \right )}

\newcommand{\abs}[1]{\left | #1 \right |}

\renewcommand{\exp}[1]{\operatorname{exp} \lrb{#1}}

\renewcommand{\log}[1]{\operatorname{log} \lrb{#1}}

\DeclareMathOperator{\erfc}{erfc}

\DeclareMathOperator{\comb}{comb}

\DeclareMathOperator{\round}{round}

\newcommand*{\fancyrefthmlabelprefix}{thm}
\newcommand*{\fancyreflemlabelprefix}{lem}
\newcommand*{\fancyrefcorlabelprefix}{cor}
\newcommand*{\fancyrefdefilabelprefix}{defi}
\frefformat{plain}{\fancyreflemlabelprefix}{lemma\fancyrefdefaultspacing#1}
\Frefformat{plain}{\fancyreflemlabelprefix}{Lemma\fancyrefdefaultspacing#1}
\frefformat{plain}{\fancyrefthmlabelprefix}{theorem\fancyrefdefaultspacing#1}
\Frefformat{plain}{\fancyrefthmlabelprefix}{Theorem\fancyrefdefaultspacing#1}
\frefformat{plain}{\fancyrefcorlabelprefix}{corollary\fancyrefdefaultspacing#1}
\Frefformat{plain}{\fancyrefcorlabelprefix}{Corollary\fancyrefdefaultspacing#1}
\frefformat{plain}{\fancyrefdefilabelprefix}{definition\fancyrefdefaultspacing#1}
\Frefformat{plain}{\fancyrefdefilabelprefix}{Definition\fancyrefdefaultspacing#1}

\algrenewcommand\algorithmicrequire{\textbf{Input:}}
\algrenewcommand\algorithmicensure{\textbf{Output:}}
\algnewcommand\AND{\textbf{and}}
\algnewcommand\OR{\textbf{or}}

\newcommand{\nat}{Nature}
\newcommand{\prl}{PRL}
\newcommand{\pra}{PRA}

\title{Low-depth Gaussian State Energy Estimation}

\author[1]{Gumaro Rendon}
\author[2]{Peter D. Johnson}

\affil[1]{Error Corp, College Park, MD 20740}
\affil[2]{Zapata Computing Inc., Boston, MA 02110 USA}
\date{\today}

\begin{document}

\maketitle

\begin{abstract}
Recent progress in quantum computing is paving the way for the realization of early fault-tolerant quantum computers.
To maximize the utility of these devices, it is important to develop quantum algorithms that match their capabilities and limitations.
Motivated by this, recent work has developed low-depth quantum algorithms for ground state energy estimation (GSEE), an important subroutine in quantum chemistry and materials.
We detail a new GSEE algorithm which, like recent work, uses a number of operations scaling as $O(1/\Delta)$ as opposed to the typical $O(1/\epsilon)$, at the cost of an increase in the number of circuit repetitions from $O(1)$ to $O(1/\epsilon^2)$.
The relevant features of this algorithm come about from using a Gaussian window, which exponentially reduces contamination from excited states over the simplest GSEE algorithm based on the Quantum Fourier Transform (QFT).
We adapt this algorithm to interpolate between the low-depth and full-depth regime by replacing $\Delta$ with anything between $\Delta$ and $\epsilon$. 
At the cost of increasing the number of ancilla qubits from $1$ to $O(\log\Delta)$, our method reduces the upper bound on the number of circuit repetitions by a factor of four compared to previous methods. 
\end{abstract}

\section{Introduction}

In recent years, the field of quantum computing has made rapid progress, from the declaration of quantum supremacy \cite{google_sup} to running the first quantum computations that challenge the best classical methods \cite{ibm_challenge}. This has been primarily driven by the quest to realize the first useful quantum computations.
One candidate for the first useful quantum computations is the task of ground state energy estimation (GSEE) \cite{aspuru2005simulated,cao2019quantum}. In the domains of chemistry and materials science, GSEE serves as a foundational subroutine for the design and discovery of new molecules and materials. The realization of quantum computations for GSEE could lead to significant advances in these fields.

Quantum phase estimation (QPE) \cite{kitaev1995quantum} has remained the standard approach to estimating ground state energy on fault-tolerant quantum computers.
A modern approach is to run QPE on the quantum walk operator block encoding of the Hamiltonian \cite{poulin2018quantum,berry2019qubitization}.
However, 
a series of recent work \cite{kim2022fault,goings2022reliably} has predicted that the quantum computers needed to run this version of QPE for classically intractable instances would be far beyond the capabilities of near-term devices.
This motivates the development of alternative approaches to GSEE that can be run earlier quantum computers.

One approach to achieving the first useful quantum computations in chemistry and materials has been to minimize the quantum resources of the algorithm at the cost of runtime and a loss of performance guarantees.
The primary algorithm developed with this approach has been the variational quantum eigensolver (VQE) algorithm \cite{peruzzo2014variational}.
However, recent findings have highlighted various challenges encountered by VQE, which pose limitations to its performance and wider application. 
These challenges include high sample costs \cite{gonthier2020identifying, johnson2022reducing} and the complexity of heuristic optimization \cite{daniel_optimization}.

These insights suggest that the pursuit of the first useful realizations of GSEE should involve not just minimizing the quantum resources required, but also adopting strategies that have reliable performance. 
Generally, the price paid for such reliability is an increase in the number of operations per circuit.
As an example, consider a quantum chemistry Hamiltonian that has been encoded into one hundred qubits, which is on the scale of classically-intractable.
Using VQE, this instance would require one-hundred physical qubits and might use a VQE ansatz with just tens to hundreds of thousands of operations.
In contrast, to run QPE as described above, even just the circuit used to encode this Hamiltonian into a quantum operation using a state of the art method \cite{von2021quantum} requires hundreds of thousands of operations and thousands of qubits \cite{kim2022fault}; and this encoding operation must be used millions of times per circuit. 
Such findings have motivated the development of reliable quantum algorithms for GSEE that reduce the required quantum resources, making them implementable on earlier quantum computers.
Such algorithms are aimed at running on so-called early fault-tolerant quantum computers (EFTQC) \cite{lin2022heisenberg}.

Recently, in Ref.~\cite{Wang_2022} the authors developed the Gaussian filter algorithm for ground-state energy estimation which affords a reduction in the number of operations per circuit.
This translates into a reduction in the number of physical qubits needed to implement the quantum computation.
The Gaussian filter GSEE algorithm uses circuit-depths of $T=\tilde O(\Delta^{-1+\alpha}\epsilon^{-\alpha})$, where $\alpha$ can be chosen between 0 and 1, interpolating between low and high-depth circuits. Multiplying this circuit depth by the number of circuit repetitions, the runtime cost scales as $MT=\tilde O(\gamma^{-4}\epsilon^{-2+\alpha}\Delta^{1-\alpha})$, where $\gamma$ is the ground state overlap with the ansatz wavefunction $|\langle\phi_{GS}|\psi\rangle|$. Afterwards, Ref.~\cite{LinLin_2023} found an alternative method with similar cost and circuit-depth scaling results (though some of the circuit depth there is spent on state preparation).
In these methods, if the overlap with the ground state, $\gamma$, is very small, the $\tilde O(1/\gamma^4)$ scaling is unfavorable. 
Motivated by this, Ref.~\cite{Wang2023} improved the cost scaling to $MT=\tilde O(\gamma^{-2}\epsilon^{-2+\alpha}\Delta^{1-\alpha})$ using a new approach combining quantum signal~\cite{QuantumSignal_Low} processing and rejection sampling. 

A favorable feature of the above methods is that they enable tuning the circuit depth (or number of operations) to accommodate the limitations of the quantum hardware. The cost of circuit depth reduction, however, is an increase in the number of circuit repetitions. Ideally, as the circuit depth is changed from $O(1/\Delta)$ to $O(1/\epsilon)$, the algorithm runtime should approach the performance of existing algorithms such as QPE that have $O(1/\epsilon)$ circuit depth scaling.
Unfortunately, previous tunable-depth algorithms do not achieve this.
Rather, in the $O(1/\epsilon)$ circuit depth setting ($\alpha=1$), they require four times more circuit repetitions than standard QPE\footnote{Setting $\alpha = 1$, $\eta = 1$ in the algorithm of \cite{Wang_2022}, the number of samples is
$\frac{128\log{\frac{4}{5\delta}}}{\pi}\approx 41\log{\frac{4}{5\delta}}$, whereas standard QPE requires only $\frac{2\log{\frac{1}{\delta}}}{(\sqrt{2}-1)^2}\approx 12\log{\frac{1}{\delta}}$ (See \Cref{lem:QPE}).
}.

In this work, we present a low-depth ground state energy estimation algorithm that achieves the following:
\begin{itemize}
    \item The circuit depth can be set by parameter $\alpha\in[0,1]$ to accommodate the device capabilities, interpolating between $O(1/\Delta_{\textup{true}})$ and $O(1/\epsilon)$.
    \item The number of circuit repetitions is reduced compared to the previous low-depth GSEE algorithm \cite{Wang_2022} by a factor of four, recovering the performance of QPE in the $(\alpha=1)$ regime.
    \item The circuit depth is reduced compared to \cite{Wang_2022} by a factor of two.
    \item The scaling of circuit repetitions with respect to ground state overlap $\gamma$ has a square-root speedup relative to \cite{Wang_2022} (achieving the same $\tilde{O}(\gamma^{-2})$ scaling as \cite{Wang2023})
 \ref{thm:short_sampling_algorithm}.
\end{itemize}
The additional cost of our algorithm is that, while previous algorithms \cite{Wang_2022, Wang2023, LinLin_2023} required only $O(1)$ ancillary qubits, ours requires $O(\log \Delta)$ extra ancillary qubits.
This trade-off may be favorable in cases where there is an overall reduction in the number of gates.
Determining when this occurs requires counting the number of gates used by each method, which depends on the implementation details of the quantum subroutines (e.g. which block encoding method or Trotter algorithm is used).

This work is organized as follows: In the next section we will provide a high-level overview of the method, following that, in \Cref{sec:algo_results}, we will will analyze the resources required to generate a single sample from a Gaussian distribution given a target error rate, as well as the maximum algorithmic error allowed on the estimated $m_{\rm th}$ moment's central value, which is summarized in \Cref{thm:short_sampling_algorithm}. In that same section, we detail in \Cref{thm:ggsee} and its proof how we use this algorithm to estimate the ground state energy within a confidence interval $\epsilon$.

\section{Overview of Methods}

Here, we provide a high-level overview of the methods used and what modifications we make over the traditional or textbook version of the quantum phase estimation algorithm. For comparison, we provide asymptotic scalings with respect to important parameters $\Delta$, upper bound on the spectral gap, $\gamma$, the overlap with the ground state, and $\epsilon$, the target precision. This will showcase the importance of the algorithm developed here, without going through the more detailed proofs.

The alternative GSEE method we propose here, consists on a modification to the standard quantum phase estimation algorithm\cite{QFT_QPE} based on the quantum Fourier transform.
The standard quantum phase estimation algorithm can be described succinctly in the language of signal processing.
Controlled-time evolution circuits are applied to an input state
$\ket{+}^{\otimes \log{T}}\ket{\psi}$ in order to encode the time signal $\exp{iE_jt}$ into the phases of ancilla states $\ket{t}\ket{\psi_j}\rightarrow e^{iE_j t}\ket{t}\ket{\psi_j}$.
Applying the quantum Fourier transform to the ancilla converts the ancilla amplitudes into the windowed Fourier coefficients (roughly a sinc function with width proportional to $\tilde O(1/T)$), which peak around the values $E_j$.
Measurements on the ancilla return bit strings with probabilities given by the squared modulus of the corresponding discrete Fourier coefficient.
We refer to these square moduli of the coefficients as the ``Fourier spectrum'' of the input signal.
The quantum computer is able to draw samples from this distribution.
These samples can be used to estimate properties of the signal, such as the frequency of the lowest-frequency component, which encodes the ground state energy.

In the case that the initial state is an eigenstate of the time-evolution operator, then the spectrum has a single peak around the encoded eigenvalue.
A simple algorithm for estimating the location of this peak is to use the fact that the peak is (roughly) the mean of the distribution, such that each sample gives (roughly) an unbiased estimate of this mean.
The variance of this estimate is the variance of the distribution, which in this case is proportional to the square of the width of the squared sinc function, or $\tilde O(1/T^2)$.
With $M$ samples, the estimate concentrates around the mean with variance $\tilde O(1/MT^2)$.
Therefore, to achieve an estimate of accuracy $\epsilon$ with high probability, we must take $M = \tilde O(1/T^2\epsilon^2)$ samples.
Each sample costs time $T$, so the total time is
$MT = \tilde O(1/T\epsilon^2)$.
To achieve accuracy $\epsilon$ with high probability using a minimal number of samples $M= \tilde O(1)$, requires setting $T= \tilde O(1/\epsilon)$.

In the general case, the spectrum will contain a mixture of many peaks, with each height proportional to the squared overlap of the input state with the corresponding eigenstate.
If our task is to estimate the lowest eigenvalue, then choosing the mean of the sampled energies will not work.
The standard strategy\cite{Ge_2017,Berry_2009_SemiQFT} is to increase the time window to $T= \tilde O(1/\gamma^2\epsilon)$, where $\gamma$ is the overlap with the ground state, take $\tilde O(1/\gamma^2)$ samples, and output the minimum energy found.
The total runtime of this algorithm, needed to estimate the ground state energy to within $\epsilon$ with high probability, is $\tilde O(1/\gamma^2\epsilon)$, and the maximum circuit depth is $T= \tilde O(1/\gamma^2\epsilon)$.

The maximum circuit depth in the above algorithms is inversely proportional to $\epsilon$.
Just like the algorithm introduced in this manuscript, if we have a promise on the energy gap $\Delta$, we can reduce this maximum circuit depth. As before, this comes at the cost of an increase in the number of samples.
Using the QFT, however, the sampling overhead in terms of $\epsilon$ is more favorable. The cost of the quantum Fourier transform, however, is an increased number of qubits and circuit depth.

We can modify the standard QPE algorithm to yield a low-depth version of QPE as follows.
Assume that the energy gap is lower-bounded by $\Delta$.
Instead of choosing the circuit depth to be $\tilde O(1/\epsilon)$, we can instead choose it to be $\tilde O(1/\Delta)$, while keeping the number of ancilla qubits to be $O(\log{1/\epsilon})$.
Now, samples from the resulting QPE circuit will be bit strings drawn from the spectrum originating from a time signal that is windowed to time $\tilde O(1/\Delta)$.

The result is that the probability distribution peaks about each eigenvalue are broadened to a width of $\tilde O(\Delta)$. Without any other state contamination, we could recover the $\tilde O(\epsilon)$ precision by taking $\tilde O(\Delta^2/\epsilon^2)$ samples.

However, with the standard QPE, the peak becomes shifted due to interference with nearby states. That is why we devised a modified method with a distribution that decays exponentially away from the window thus eliminating significant contamination from nearby states.

Given this spectral distribution, we are able to select samples lying within a spectral window of width $\tilde O(\Delta)$ for which excited state contributions are going to be negligible. In order to have this isolated window, we need $T=\tilde O(1/\Delta)$. Now, to ensure that we get at least one outcome from within this target window we expect to take $\tilde O(1/\gamma^2)$ samples. So the number of kept samples is $\tilde O(M\gamma^2)$ out of the total samples $M$. With this, the sample mean concentrates around the true mean with variance $\tilde O(1/T^2M\gamma^2)=\tilde O(\Delta^2/M\gamma^2)$.
To achieve $\epsilon$ accuracy w.h.p., we set $M=\tilde O(\Delta^2/\epsilon^2\gamma^2)$.
The total time is $MT=\tilde O(\Delta/\epsilon^2\gamma^2)$ and the maximum circuit depth is $T=\tilde O(1/\Delta)$.

Although we have set-up the algorithm to work in the low-depth regime, $T=\tilde O(1/\Delta)$, we can always choose to replace the algorithm's input parameter $\Delta$ with anything in $[\epsilon,\Delta]$, thus we could take full advantage of the quantum speed when $\Delta \to \epsilon$.

\section{Algorithm and Results\label{sec:algo_results}}

We now state our results and describe the ground state energy estimation algorithm. The following theorem gives the performance of Algorithm \ref{alg:sampl_n_trim}. We state the theorem generally in terms of estimating any moment of a Gaussian (the mean is the first moment), although we do not develop or propose any applications that use this generalization (beyond ground state energy estimation). Later, we will specialize this algorithm for the purpose of ground state energy estimation in our main result (see Theorem \ref{thm:ggsee}).
\begin{theorem}\label{thm:short_sampling_algorithm}
There exists an algorithm to sample the $m_{\rm th}$ moment, where $m\geq 1$, of a random variable $x$ which is approximated by the normal distribution $\mathcal{N}(\tilde\sigma,\mu=E_0)$. Here, $E_0$ is the ground state energy of $H$, and $\|H\|\leq 1-\Delta$. $\Delta$ is the lower bound on the spectral gap and $\eta\leq\abs{\gamma_0}^2$ is the lower bound on the initial state squared overlap with the ground state.
To obtain at least one (approximate) sample of the $m_{\rm th}$ moment, a sufficient number of repetitions of the circuit in \Cref{fig:gaussian_qpea} is
    $$
    M_0= O\left(\frac{\log{1/\delta}}{\eta}\right),
    $$
where $\delta$ is the tolerable failure rate of the algorithm. Also, asymptotically as $\Delta \to 0$, we require $\tilde\sigma = \tilde O(\Delta/\sqrt{m})$ and a circuit depth $\tilde O\left(\frac{m}{\Delta}\right)$.
\end{theorem}

\begin{algorithm}
\caption{Sample from Gaussian around GS Eigenvalue}\label{alg:sampl_n_trim}
\begin{algorithmic}
\Require $\delta_{\rm max}\leq 0.01$, $\eta$, $\Delta_{\rm max}$, $m$, $\tilde\varepsilon_m$
\Ensure Samples that approximate moment $m$ of normally distributed random variable with error up to $\tilde\varepsilon_m$
\State $\Delta \gets \Delta_{\rm max} $
\State $\delta \gets \delta_{\rm max}$ 
\State $M_0\equiv\left\lceil\frac{16}{3 \eta}\log{\frac{3}{ \delta }}\right\rceil$
\State $\tilde\sigma \equiv \left( \frac{1}{\sqrt{m}} \frac{ \Delta/6 }{\sqrt{2 \log{\frac{4 M_0}{\delta} \left( 1 + \sqrt{\frac{5}{3}}\sqrt{\frac{1-\eta}{\eta}}\right)^2} } }\right)$
\State $q \equiv \left\lceil {\rm log}_2\left(\frac{3 m}{\pi \Delta}\sqrt{1+3\log{\frac{m}{4 e \pi^2 \tilde\sigma^2 }\left(\frac{(m!)^2 C(\eta)}{\tilde\varepsilon_m}\right)^{2/m}}} \sqrt{2 \log{\frac{4 M_0}{\delta} \left( 1 + \sqrt{\frac{5}{3}}\sqrt{\frac{1-\eta}{\eta}}\right)^2} }\right)\right\rceil$
\State Decrease $\Delta$ such that $\abs{G_0(0)-\Tilde{G}_0} \leq 1/8$, $\abs{ \tilde{G}_0 - \tilde{F}_0 } \leq 1/8$, $\| \epsilon^{(R)} \|/\sqrt{\eta} \leq 1/8$, and $\frac{m}{2 e \pi^2 \tilde\sigma^2 }\left(\frac{(m!)^2 C(\eta)}{\tilde\varepsilon_m}\right)^{2/m} > e$ (While updating $q$ and $\tilde\sigma$)

\State Decrease $\delta$ such that $\frac{\sqrt{ \log{\frac{4 M_0}{\delta} \left( 1 + \sqrt{\frac{5}{3}}\sqrt{\frac{1-\eta}{\eta}}\right)^2} }}{\sqrt{1+3\log{\frac{m}{4 e \pi^2 \tilde\sigma^2 }\left(\frac{(m!)^2 C(\eta)}{\tilde\varepsilon_m}\right)^{2/m}}}} \geq 2$ (While updating $M_0$, $\tilde\sigma$, and $q$)
\State ${\rm basket} \gets \{\}$
\State ${\rm basket_{aux}} \gets \{\}$
\For{$j \gets 1$ to $M_0$}
    \State Draw sample and add to ${\rm basket_{aux}}$
\EndFor
\State Take lowest key with non-zero value from ${\rm basket_{aux}}$ and add it to basket
\State Also add to basket all the outcomes in ${\rm basket_{aux}}$ that are $2K =  \lfloor (2/3) 2^q \Delta \rfloor-1$ to the right of lowest
\State Return basket
\end{algorithmic}
\end{algorithm}

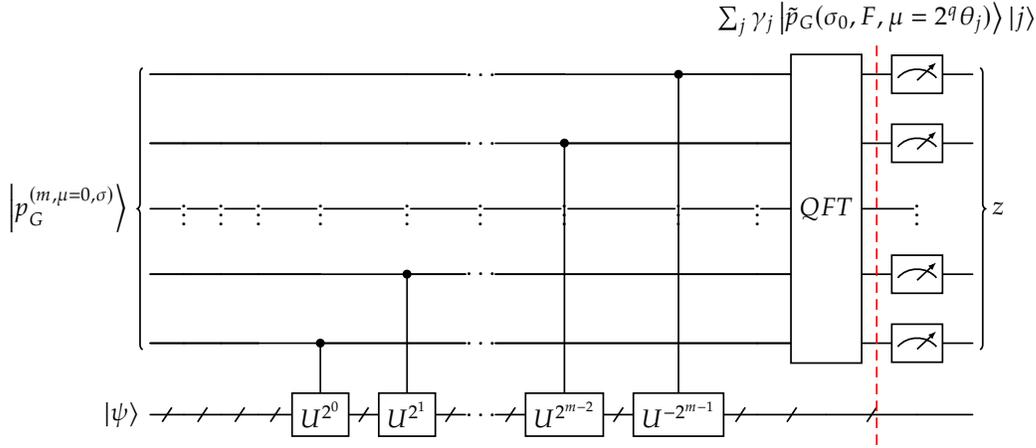
\begin{figure*}[t!]
\scalebox{0.8}{
\begin{quantikz}
 \lstick[wires=5]{$\ket{p^{(m,\mu=0,\sigma)}_{G}} $} & \qw & \qw &  \qw & \qw & \qw & \qw\ldots &\qw & \ctrl{5}  & \qw & \gate[wires=5,nwires=3]{QFT}\qw\slice{$\sum_j \gamma_j \ket{\tilde{p}_G ( \sigma_0,F,\mu=2^q\theta_j)}\ket{j}$} & \meter{} & \rstick[wires=5]{$z$}
 \\
  & \qw &\qw &\qw& \qw&\qw & \qw\ldots & \ctrl{4}   & \qw & \qw & \qw &  \meter{} &
 \\
  & \vdots & \vdots &\vdots& \vdots & \vdots & \vdots &\vdots&\vdots&\vdots&\vdots& \vdots
 \\
  & \qw &\qw &\qw& \qw&\ctrl{2}  & \qw\ldots & \qw   & \qw & \qw & \qw &  \meter{} &
 \\
  & \qw & \qw&\qw& \ctrl{1} & \qw  & \qw\ldots & \qw  & \qw & \qw &  \qw & \meter{} &
 \\
 \lstick{$\ket{\psi}$} & \qwbundle[alternate]{} & \qwbundle[alternate]{} & \qwbundle[alternate]{} & \gate{U^{2^0}}\qwbundle[alternate]{}& \gate{U^{2^1}}\qwbundle[alternate]{} & \qwbundle[alternate]{}\ldots & \gate{U^{2^{m-2}}}\qwbundle[alternate]{} & \gate{U^{-2^{m-1}}}\qwbundle[alternate]{} & \qwbundle[alternate]{} & \qwbundle[alternate]{} & \qwbundle[alternate]{} &
\end{quantikz}
}
\caption{Circuit to implement a Gaussian $m$-qubit phase estimation algorithm. Here, $U=\exp{i2\pi H}$, the evolution operator for the target Hamiltonian $H$.}\label{fig:gaussian_qpea}
\end{figure*}

The algorithm described above gives us a random variable whose mean is close to $\tilde\mu_1=E_0$ when specializing to $m=1$.
Therefore, if we construct an estimator for this mean from a sample average over sufficiently many samples, then we can ensure that this estimator will be close to $E_0$ with high probability.

We give a useful statistical lemma that we can combine with the previous results to establish our main result.

\begin{lemma}
\label{lem:stat_error}
Assume that $x\in[a,a+b]$ is a random variable such that $|\mathbb{E}(x)-\theta|<c\epsilon$.
Then for any $c< 1$ the random variable $x_M = \frac{1}{M}\sum_{i=1}^{M}x_i$, constructed from $M$ independent draws of random variable $x$, satisfies $|x_M-\theta|\leq \epsilon$ with probability greater than $1-\delta$ when
\begin{align}
    M\geq \frac{b^2}{2(\epsilon-c\epsilon)^2}\ln{\frac{2}{\delta}}.
\end{align}
\end{lemma}
\begin{proof}
We first upper bound the failure probability using a triangle inequality:
\begin{align}
    \textup{Pr}(|x_M-\theta|>\epsilon) &\leq \textup{Pr}(|x_M-\mathbb{E}(x)|>\epsilon-|\mathbb{E}(x)-\theta|)\\\nonumber
    & =\textup{Pr}(|x_M-\mathbb{E}(x)|>\epsilon-c\epsilon)
\end{align}
By the Hoeffding bound, the probability that $x_M$ deviates by more than $\epsilon-c\epsilon$ from its mean is upper bounded as
\begin{align}
\textup{Pr}(|x_M-\mathbb{E}(x)|>\epsilon-c\epsilon)\leq 2\exp{-2M(\epsilon-c\epsilon)^2/b^2}.  
\end{align}
From this, we have that with probability greater than $1-\delta$, $x_M$ is within $\epsilon$ of $\theta$ when
\begin{align}
    M\geq \frac{b^2}{2(\epsilon-c\epsilon)^2}\ln{\frac{2}{\delta}}.
\end{align}
\end{proof}
\noindent In the above result, we see that when we are able to design an estimator for which $c=O(1)<1$, then the sample complexity is $O(b^2/\epsilon^2)$.

We can now state the main result.
\begin{theorem}\label{thm:ggsee}
\Cref{alg:ggsee} provides an estimate, $\hat\mu_1$, of $\tilde\mu_1$ within an error $\epsilon$ using $M=\left\lceil \frac{8\Delta^2}{9\epsilon^2(1-c)^2} \log{\frac{4}{\delta}}\right\rceil$ calls of \Cref{alg:sampl_n_trim} with a total error rate that is $\delta$, provided that the error rate for \Cref{alg:sampl_n_trim} is at most $\frac{\delta}{4 M}$ and its target error is $c\epsilon$, where $c=O(1)\leq 1$.  
\end{theorem}
\begin{proof}
Algorithm \ref{alg:ggsee} fails when $|\hat\mu_1-\tilde\mu_1|>\epsilon$ given underlying \Cref{alg:sampl_n_trim} does not fail or when \Cref{alg:sampl_n_trim} fails at least once.
The underlying sampling algorithm, \Cref{alg:sampl_n_trim}, involves post-selecting only those outcomes $x$ satisfying $x\in[x_{\textup{min}},x_{\textup{min}}+b]$. We know that $b\leq4\Delta/3$ from \Cref{thm:final_sampling_algorithm}. 
The probability of failure of the higher level algorithm, \Cref{alg:ggsee}, is upper bounded as
\begin{align}
    \textup{Pr}(\textup{Alg. 2 fails}) &\leq \textup{Pr}(|\hat\mu_1-\tilde\mu_1|>\epsilon \mid \textup{Alg. 1 does not fail throughout}) \cr 
    &+ \textup{Pr}(\textup{Alg. 1 fails at least once})\cr
    & = \delta_2 + (1 - (1-\tilde{\delta}_1)^M) \cr
    & = \delta_2 + 1 - \exp{\log{1-\tilde{\delta}_1}^M} \cr
    & \leq \delta_2 - M\log{1-\tilde{\delta}_1} \cr
    & = \delta_2 + M\log{\frac{1}{1-\tilde{\delta}_1}} \cr
\end{align}
Using that
\begin{align}
    \log{\frac{1}{1-\tilde\delta_1}} \leq \frac{\tilde\delta_1}{(1-\tilde\delta_1)} 
\end{align}
for $0\leq \tilde\delta_1\leq 1$,
then, provided that $\tilde\delta_1\leq 1/2$, we can simplify to
\begin{align}
    \log{\frac{1}{1-\tilde\delta_1}} \leq 2\tilde\delta_1. 
\end{align}
With this, we obtain
\begin{align}
    \textup{Pr}(\textup{Alg. 2 fails}) & \leq \delta_2 + 2M\tilde\delta_1.
\end{align}
We now identify $\delta_2 = \delta/2$ and $\tilde\delta_1 = \frac{\delta}{4 M}$.

Thus, every time we run \Cref{alg:sampl_n_trim} we pass it the target minimum error rate $\tilde\delta_1 = \frac{\delta}{4 M}$, while $\delta_2=\delta/2$ is passed to the procedure of \Cref{lem:stat_error} such that the overall failure rate of \Cref{alg:ggsee} is bounded from above by $\delta$. 
By \Cref{thm:final_sampling_algorithm,lem:stat_error}, Algorithm \ref{alg:ggsee} can use $\tilde{O}(\frac{\Delta^2}{\eta\epsilon^2})$ samples from quantum circuits with expected depth $\tilde{O}(1/\Delta)$ to generate a random variable that has mean within $c\epsilon$ of $\tilde\mu_1$, where $c=O(1)\leq 1$.
\end{proof}

\begin{algorithm}
\caption{Gaussian Ground State Energy Estimation}\label{alg:ggsee}
\begin{algorithmic}
\Require $\delta$, $\eta$, $\Delta$, $\tilde\sigma_{\rm max}$, $m=1$, $\epsilon$, $c$
\Ensure $\hat\mu_1$ which is within $\epsilon$ of $\tilde\mu_{1}$
\State $\delta_2\gets \frac{\delta}{2}$
\State $b\gets \frac{4\Delta}{3}$
\State $M\gets \left\lceil \frac{b^2}{2\epsilon^2(1-c)^2} \log{\frac{2}{\delta_2}}\right\rceil=\left\lceil \frac{8\Delta^2}{9\epsilon^2(1-c)^2} \log{\frac{4}{\delta}}\right\rceil$
\State $\tilde\delta_1\gets \frac{\delta}{4 M}$
\For{$j\gets 1$ to M}
\State Pass arguments  $\tilde\delta_1$, $\eta$, $\Delta$, $m=1$, $c\epsilon$ to \Cref{alg:sampl_n_trim} to obtain at least one sample of $x$.
\EndFor
\State Return $\hat\mu_1$
\end{algorithmic}
\end{algorithm}

We now introduce a corollary which tells us how we bridge the two limits of application of \Cref{alg:ggsee}: shallow depth and full depth. This also illustrates the trade-off in number of samples as we go to shallower depths. Here, $\Delta_{\rm true}$ is the true spectral gap of the target Hamiltonian and $\epsilon$ remains the target tolerance on the estimated ground state. Thus, we parameterize $\Delta = \Delta_{\rm true}^{1-\alpha}\epsilon^{\alpha}$ such that we get the shallow limit with $\alpha =0$ and the full-depth limit when $\alpha =1 $.

\begin{corollary}\label{cor:interp}
Given that we fix $c=1 - 2 \sqrt{2}/3$ and parametrize $\Delta$ using $\Delta = \Delta_{\rm true}^{1-\alpha}\epsilon^{\alpha}$, \Cref{alg:ggsee} provides an estimate, $\hat\mu_1$, of $\tilde\mu_1$ within an error $\epsilon$ using $M M_0$ calls of the circuit in \Cref{fig:gaussian_qpea}, where
$$
M=\left\lceil \epsilon^{-2+2\alpha} \Delta_{\rm true}^{2-2\alpha}  \log{\frac{4}{\delta}}\right\rceil$$
and
$$
M_0=\left\lceil\frac{16}{3 \eta}\log{\frac{12 \left\lceil \epsilon^{-2+2\alpha} \Delta_{\rm true}^{2-2\alpha} \log{\frac{4}{\delta}}\right\rceil }{ \delta }}\right\rceil,
$$
 with a total error rate that is $\delta$, where each circuit uses a number of queries of $c$-$U$ of at most
\begin{align*}
    2^q &\leq\Delta_{\rm true}^{-1+\alpha}\epsilon^{-\alpha}\frac{6\sqrt{2}}{\pi}\sqrt{1+3\log{ \frac{18}{e \pi^2 \Delta^2 } \left(\frac{ 3 C(\eta)}{(3-2\sqrt{2})\epsilon}\right)^{2}\log{\frac{4 M_0}{\delta} \left( 1 + \sqrt{\frac{5}{3}}\sqrt{\frac{1-\eta}{\eta}}\right)^2}}} \cr 
    &\qquad\times\sqrt{\log{\frac{4 M_0}{\delta} \left( 1 + \sqrt{\frac{5}{3}}\sqrt{\frac{1-\eta}{\eta}}\right)^2} },
\end{align*}
where
\begin{align}
    C(\eta) &= \left[\left(\frac{128}{45}\right)\exp{12}  \left(12 + 3 + \frac{9}{4}\sqrt{\frac{1}{\eta}}   \right) \right.\cr
    &\left. +10\exp{12} + \frac{55}{8} \exp{2} \left(  1  + \sqrt{\frac{5}{3}}\sqrt{\frac{1}{\eta}} \right)\right].
\end{align}
\end{corollary}
For the purpose of comparison, we build off of the analysis in \cite{Cleve_1998} to compare the constant factors of our algorithm to that of quantum phase estimation.
\begin{lemma}\label{lem:QPE}
The traditional quantum phase estimation algorithm~\cite{Cleve_1998} can be used to estimate an eigenphase of $c$-$U$ to within error $\epsilon$ using a register with $q=\left\lceil{\rm log_2}(1/\epsilon)\right\rceil$ ancilla qubits, a query depth of $2^q$, and samples
   \begin{align*}
    n \geq \frac{2}{(\sqrt{2}-1)^2}\log{\frac{1}{\delta}},
\end{align*}
where $\delta$ is the tolerable failure rate.
\end{lemma}

\begin{proof}
For a tolerable failure rate $\delta$, the traditional quantum phase estimation algorithm \cite{Cleve_1998} uses
\begin{align}
    m =  \left\lceil {\rm log}_2\left(\frac{1}{2 \varepsilon}\right) + \frac{1}{2} \right\rceil
\end{align}
additional ancilla qubits beyond those needed to ensure an $\epsilon$-accurate estimate.
Rather than increasing the number of ancilla qubits, we could alternatively increase the number of samples and carry out a majority vote on the outcomes. To find the minimal quantum resources needed to ensure no more than a $\delta$ failure rate, we set $m=1$ and bound the failure rate of the majority vote strategy.
With $m=1$, a single sample has failure rate
\begin{align}
    \delta \leq \frac1{2\sqrt{2}}.
\end{align}
yielding a success rate of
\begin{align}
    p_{\rm success} \geq 1-\frac{1}{2\sqrt{2}} \approx 0.6464.
\end{align}
For even $n$, a tail bound given by Hoeffding's inequality for a binomial random variable $X$ is:
\begin{align}
    \delta = P(X\leq n/2)&=F(k=n/2;n,p_{\rm success}) \leq \exp{-2 n\left(p_{\rm success}-1/2\right)^2} \cr
    & \leq \exp{- \frac{n}{2}\left(\sqrt{2}-1)\right)^2}
\end{align}
Solving for $n$
\begin{align}
    n \geq \frac{2}{(\sqrt{2}-1)^2}\log{\frac{1}{\delta}}.
\end{align}
\end{proof}
This analysis gives constant factors that show that (according to the established upper bounds) the full-depth version of the Gaussian phase estimation algorithm, which achieves a factor of $16/3\approx 5.33$ in the $\eta\rightarrow 1$ limit, ``recovers'' (as in, does not exceed) the sample counts required by the traditional quantum phase estimation algorithm, which is roughly $11.66$.
This shows that the Gaussian phase estimation algorithm enables a smooth transition in the sample complexity from low depth to high depth.

\section{Discussion and Outlook}

We have detailed a new algorithm for ground state energy estimation using a modification of standard quantum phase estimation, which exponentially suppresses contiguous state contamination on the resulting spectrum. The method interpolates between low-depth, $O(1/\Delta_{\rm true})$, and full-depth, $O(1/\epsilon)$, which can be used to accommodate the capabilities of the quantum architecture. 
This feature makes the algorithm well-suited to the era of early fault-tolerant quantum computing.

Our work finds an improvement with respect to~\cite{Wang_2022} with a factor of four reduction in the number of samples and a factor of two reduction in the depth (in the low-depth regime). Moreover, it achieves a square-root improvement relative to \cite{Wang_2022} in terms of the overlap scaling, achieving the same $\tilde O(\gamma^{-2})$ dependence established in \cite{Wang2023}. 
The cost scaling with respect to important parameters is comparable to other state-of-the-art alternatives (e.g. \cite{Wang2023},\cite{LinLin_2023}). Constant factors in the analytical upper bounds depend on the proof technique and could be overly pessimistic. 
Although tighter analytic bounds for all methods could be achieved in principle, we suggest that numerical investigations be used determine which method is most advantageous.

For future work, we anticipate that the generalization to estimating higher-order moments may find application or that the analysis may be of independent interest. 
It would also be important to investigate methods to reduce the number of ancilla qubits used by the algorithm, similar to \cite{najafi2023optimum}.
We hope that our methods might eventually make other algorithms that rely on QFT, such as \cite{shor1994algorithms} or estimation of multiple expectation values \cite{Multi_Est_Huggins_2022}), more amenable to implementation on early fault-tolerant quantum computers.

\section{Acknowledgements}

Authors would like to acknowledge the useful comments by Guoming Wang and Daniel Stilck Fran\c{c}a. Gumaro R. was supported by Zapata Computing during part of the manuscript writing.

\bibliographystyle{unsrt}

\appendix

\section{Proof of \Cref{thm:final_sampling_algorithm}, \Cref{alg:sampl_n_trim}'s main theorem}

Let us start with the outcome of the ancillary register of the Gaussian Phase Estimation (GPE) (See \Cref{fig:gaussian_qpea})\cite{Trotter_extrapol}
\begin{align}\label{eq:approx_gauss}
\sum_j \gamma_j \ket{\tilde{p}_G ( \sigma_0,F,\mu=2^q\theta_j)}\ket{j},
\end{align}
where $\ket{j}$ are eigenstates of the Hamiltonian of interest and $\ket{\tilde{p}_G ( \sigma_0,F,\mu) }$ efficiently approximates (Theorem 12, \cite{Trotter_extrapol}) the state
\begin{align}\label{eq:state}
\frac{1}{\sqrt{\mathcal{N}} }\sum_k\sqrt{g_0( (k-\mu) \bmod 2^q,0)} \ket{k},
\end{align}
where
$$
(k-\mu) \bmod 2^q = r = (k-\mu) - 2^q \round\left(\frac{k-\mu}{2^q}\right).
$$
and
$$
g_0(x,\mu) = \frac{1}{\sigma\sqrt{2\pi}} \exp{-\frac{1}{2}(x-\mu)^2/\sigma^2}.
$$
The mod function in \Cref{eq:state} is implementing the periodicity of the functions that come out of the $\rm QFT$ in the Gaussian phase estimation circuit ~\ref{fig:gaussian_qpea}. The convention chosen for the round$(\cdot)$ function is that it rounds to nearest integer, and when at half, it rounds to the even integer.
The factor $\mathcal{N}$ is here to ensure normalization through
$$
\mathcal{N} = \sum^{2^q/2-1}_{k=-2^q/2} g_0 (k, \tilde{\mu}),
$$
where
$$
\tilde{\mu} = \mu - \round(\mu).
$$

We propose we do the GPE experiment in a round of $M_0=O(1/\eta)$ samples such that at least one outcome is from the ground state's Gaussian. Process formalized in \Cref{thm:final_sampling_algorithm}. Here, $\eta$ is a lower bound such that $\eta\leq \gamma_0^2$.

Through this process, we are interested in calculating the moments of the normal distribution:
$$
\mu_m = \int g_m(x) \mathrm{d}x,
$$
where the new definition
$$
g_m(x) = x^m g_0(x),
$$
is used.

It will be convenient to define
$$
G_m(k)=\mathcal{F}_x\left(g_m(x)\right)(k),
$$
where 
$$
\mathcal{F}_x\left(f(x)\right)(k)=\int^{\infty}_{-\infty} \exp{-2\pi i x k} f(x) dx,
$$
such that
$$
\mu_m = G_m(0),
$$
and where the following moment generating property can be used
$$
G_m(k) = (-i)^m \left(\frac1{2\pi}\right)^m \frac{\mathrm{d}^m}{\mathrm{d} k^m} G_0(k).
$$

We will bound sources of error when estimating Gaussian observables off of the ancillar register outcome. The total error on the $m$-th moment will be
\begin{align}
\varepsilon_m &= \abs{G_m(0)-\frac{1}{\mathcal{N}}\left(\frac{\abs{\gamma_0}^2}{\tilde{p}_0}\right) \tilde{F}_m} \cr
\tilde{F}_m &= \sum^{K}_{n=-K} \left|(f)_n\right|^2 n^m \cr 
f_n &=  \sqrt{g_0(n)}
\end{align}

To facilitate a clear presentation, we will take a stepped approach to estimate all the comprising sources of error. Here is the list:
\begin{itemize}
    \item Renormalization
        \begin{align}
            \varepsilon^{\rm norm}_m = \left| G_m(0) - \frac{1}{\mathcal{N}}\left(\frac{\abs{\gamma_0}^2}{\tilde{p}_0}\right) G_m(0) \right| = \abs{ 1 - \frac{1}{\mathcal{N}}\left(\frac{\abs{\gamma_0}^2}{\tilde{p}_0}\right) } \left| G_m(0) \right| 
        \end{align}
    \item Discretization
        \begin{align}
            \varepsilon^{\rm discret}_m &= \frac{1}{\mathcal{N}}\left(\frac{\abs{\gamma_0}^2}{\tilde{p}_0}\right)\left| G_m(0) - \tilde{G}_m \right| \cr
            \tilde{G}_m &= \sum^{\infty}_{n=-\infty} g_m(n)
        \end{align}    
    \item Truncation and Contamination

        \begin{align}
            \varepsilon^{\rm trunc+polut}_{m} &= \frac{1}{\mathcal{N}}\left(\frac{\abs{\gamma_0}^2}{\tilde{p}_0}\right)\left|\tilde{G}_m - \tilde{F}^{(\rm polut)}_m  \right| \cr
        \end{align}       
    \begin{align}
        \tilde{F}^{(\rm polut)}_m &= \sum^{K}_{x=-K} \left|f^{(\rm polut)}_x\right|^2 x^m \cr 
        f^{(\rm polut)} &= f + \frac{\sqrt{\mathcal{N}}}{\gamma_0}\,\epsilon^{\rm (polut)}
    \end{align}
    Thus, by triangle inequality,
    \begin{align}
        \varepsilon_m \leq \varepsilon_m^{\rm norm} + \varepsilon_m^{\rm discret} + \varepsilon_m^{\rm trunc+polut}.
    \end{align}
    We will find it convenient to use the vector norm notation $\| \cdot \|_i$ for the norms on $\mathbb{C}^{2K+1}$ vectors supported on the window of interest. If no index is used, the 2-norm is assumed.
    It will also be useful to define the relative variables
    \begin{align}
        \tilde{\sigma}  &=  \sigma/2^q. \cr
        \tilde\mu_m &= \mu_m/\left(2^{q}\right)^m \cr
        \tilde\varepsilon_m &= \varepsilon_m/\left(2^{q}\right)^m.
    \end{align}

\end{itemize}

\subsection{Bounds on Normalization Factor $\mathcal{N}$}

First, we will bound the normalization factor $\mathcal{N}$.

\begin{lemma}\label{lem:norm_error}
Let $g_0 = \frac{1}{\sigma\sqrt{2\pi}}\exp{-\frac{1}{2}(x-\mu)^2/\sigma^2}$, and $-1/2 \leq \tilde\mu < 1/2$

$$
\mathcal{N}(\sigma,2^q,\tilde\mu) = \sum^{2^q/2-1}_{n=-2^q/2} g_0 (n,\tilde\mu).
$$
Then, 
$$
\sup_{\tilde\mu}\mathcal{N}(\sigma,2^q,\mu) \leq \mathcal{N}_{\rm up} = 1+\abs{ \tilde{G}_0 - \tilde{F}_0 }+|G_0(0)-\tilde{G}_0|,
$$
and the lower bound
$$
\inf_{\tilde\mu}\mathcal{N}(\sigma,2^q,\mu) \geq \mathcal{N}_{\rm low} = 1-\abs{ \tilde{G}_0 - \tilde{F}_0 }-|G_0(0)-\tilde{G}_0|.
$$
Moreover, if $1-\abs{ \tilde{G}_0 - \tilde{F}_0 }-|G_0(0)-\tilde{G}_0|>0$, we obtain
$$
\left|1-\frac{1}{\mathcal{N}}\right| \leq
\frac{\abs{ \tilde{G}_0 - \tilde{F}_0 } + |G_0(0)-\tilde{G}_0|}{1-\abs{ \tilde{G}_0 - \tilde{F}_0 } - |G_0(0)-\tilde{G}_0|}.
$$

\end{lemma}
\begin{proof}
We have that
$$
|G_0(0)-\tilde{G}_0| = |1-\tilde{G}_0|,
$$
and we also know that
$$
|\tilde{G}_0-\mathcal{N}| = \left(\sum^{\infty}_{n=2^q/2} g_0(n) + \sum^{-2^q/2-1}_{n=-\infty} g_0(n) \right) \leq \abs{ \tilde{G}_0 - \tilde{F}_0 },
$$
where
$$
\abs{ \tilde{G}_0 - \tilde{F}_0 } = \sum^{\infty}_{n=K+1} g_0(n) + \sum^{-K-1}_{n=-\infty} g_0(n),
$$
and $K$ is an integer where $2K+1<2^q$.
Thus, by triangle inequality and knowing $G_0(0)=1$,
$$
|1-\mathcal{N}|\leq \abs{ \tilde{G}_0 - \tilde{F}_0 } + |G_0(0)-\tilde{G}_0|.
$$
We find it useful to define the following upper bound
$$
\sup_{\mu}\mathcal{N}(\sigma,2^q,\mu) \leq \mathcal{N}_{\rm up} = 1+\abs{ \tilde{G}_0 - \tilde{F}_0 }+|G_0(0)-\tilde{G}_0|
$$
and the lower bound
$$
\inf_{\mu}\mathcal{N}(\sigma,2^q,\mu) \geq \mathcal{N}_{\rm low} = 1-\abs{ \tilde{G}_0 - \tilde{F}_0 }-|G_0(0)-\tilde{G}_0|.
$$
If we assume that $1-\abs{ \tilde{G}_0 - \tilde{F}_0 }-|G_0(0)-\tilde{G}_0|>0$ we can also write the following bound:
$$
\left|1-\frac{1}{\mathcal{N}}\right| \leq \left|\frac{\mathcal{N}_{\rm up}-1}{\mathcal{N}_{\rm low}}\right| =
\frac{\abs{ \tilde{G}_0 - \tilde{F}_0 } + |G_0(0)-\tilde{G}_0|}{1-\abs{ \tilde{G}_0 - \tilde{F}_0 } - |G_0(0)-\tilde{G}_0|}.
$$
\end{proof}

To get a more transparent idea of how the bounds on the normalization factor depend on $K$ and $\sigma$, we will upper-bound $\abs{ \tilde{G}_0 - \tilde{F}_0 }$ in terms of these quantities through the lemma:

\begin{lemma}\label{lem:eps1_error}
For an integer $K$ and $-1/2\leq\tilde{\mu} < 1/2$, given the definition
$$
\abs{ \tilde{G}_0 - \tilde{F}_0 } = \sum^{\infty}_{n=K+1} g_0(n,\tilde{\mu}) + \sum^{-K-1}_{n=-\infty} g_0(n,\tilde{\mu}),
$$
we obtain the following upper bound
$$
\abs{ \tilde{G}_0 - \tilde{F}_0 } \leq \max_{\tilde{\mu}}\abs{ \tilde{G}_0 - \tilde{F}_0 } \leq \exp{-\frac{(K-1/2)^2}{2\sigma^2}}.
$$
\end{lemma}
\begin{proof}
\begin{align*}
\abs{ \tilde{G}_0 - \tilde{F}_0 } &\leq 2 \max_{\tilde{\mu}}\max\left\{\sum^{\infty}_{n=K+1} g_0(n,\tilde{\mu}),\right.\cr &\qquad\left.\sum^{-K-1}_{n=-\infty} g_0(n,\tilde{\mu})\right\} \cr
& = 2 \sum^{\infty}_{n=K+1} g_0(n,-1/2)
\end{align*}
Here, we can use the right-Riemann summation approximation to bound with the integral
\begin{align*}
\abs{ \tilde{G}_0 - \tilde{F}_0 } & \leq 2 \int^{\infty}_{K} g_0(n,-1/2) d n \cr 
&\leq \erfc\left(\frac{K-1/2}{\sqrt{2}\sigma}\right) \leq \exp{-\frac{(K-1/2)^2}{2\sigma^2}}.
\end{align*}
\end{proof}

\subsection{Bounds on $\|\epsilon^{\rm (polut)}\|$}

\begin{lemma}[Bounds on excited state contribution]\label{lem:contamination_bounds}
    
\end{lemma}

\begin{proof}

An upper bound on the norm of the excited state contamination from the right, before renormalizing by the exact hit rate $\tilde{p}_0$, would be

\begin{align}
    \lVert \epsilon^{\rm (polut)} \rVert \leq \sqrt{\frac{1-\eta}{\inf_{\tilde\mu}(\mathcal{N})}}\|\epsilon^{(R)}\|.
\end{align}
where we have defined the following ancillary vector $\left(\epsilon^{(R)}\right)_j = \sqrt{g_0(j,\tilde{\mu}+2^q\Delta )}$
\begin{align}
    \| \epsilon^{(R)} \|^2  &= \sum^{K}_{n=-K} g_0(n,\tilde{\mu}+2^q\Delta )
\end{align}
We also have another bound for $\left\Vert \epsilon^{\rm (polut)} \right\Vert$
\begin{align}
    \lVert \epsilon^{\rm (polut)} \rVert \leq \sqrt{\frac{1-\eta}{\inf_{\tilde\mu}(\mathcal{N})}}\|\epsilon^{(L)}\|,
\end{align}
where  $\left(\epsilon^{(L)}\right)_j = \sqrt{g_0(j,\tilde{\mu}-2^q\Delta')}$
\begin{align}
    \| \epsilon^{(L)} \|^2  &= \sum^{K}_{n=-K} g_0(n,\tilde{\mu}-2^q\Delta' ).
\end{align}
Here, $\Delta'$ is the spectral gap from below induced by the periodicity of the amplitudes in \Cref{eq:state}.
If we constrain the norm through
$$
\lVert H \rVert \leq \sum_{j} \lVert H_j \rVert = 1/2 - \Delta/2,
$$
we can set $\Delta' = \Delta$. Thus
\begin{align}
    \| \epsilon^{(L)} \|^2& = \sum^{K}_{n=-K} g_0(n,\tilde{\mu}-2^q\Delta ).
\end{align}
Because of the even symmetry of $g_0$ we note that
\begin{align}
    \max_{\Tilde{\mu}} \| \epsilon^{(R)} \|^2 = \max_{\Tilde{\mu}} \| \epsilon^{(L)} \|^2
\end{align}

Thus, we will focus only on bounding $\| \epsilon^{(R)} \|$. We note that
\begin{align}
      \| \epsilon^{(R)} \|^2 \leq \max_{\Tilde{\mu}} \| \epsilon^{(R)} \|^2  &\leq \max_{\Tilde{\mu}}\sum^{K}_{n=-\infty} g_0(n,\tilde{\mu}+2^q\Delta ). \cr
\end{align}
We can now use left Riemann sum to bound the summation with
\begin{align}
    \| \epsilon^{(R)} \|^2& \leq \max_{\Tilde{\mu}}\int^{K+1}_{-\infty} g_0(n,\tilde{\mu}+2^q\Delta ) \cr
    &\leq \erfc\left(\frac{2^q\Delta -K-1/2}{\sqrt{2} \sigma}\right) \leq \exp{-\frac{(2^q\Delta -K-1/2)^2}{2 \sigma^2}}.
\end{align}
\end{proof}

\subsection{Hit Rate}

We now will bound the exact hit rate $\tilde{p}_0$ through
\begin{lemma}[Hit rate]\label{lem:hit_rate}
Assuming that $\abs{G_0 (0) - \tilde{G}_0}\leq 1/8$, $\abs{ \tilde{G}_0 - \tilde{F}_0 } \leq 1/8$, and $\sqrt{\frac{1}{\eta}} \| \epsilon^{(R)} \| \leq 1/8$, we have that
\begin{align*}
    \tilde{p}_0 &\geq \abs{\gamma_0}^2\left(\frac{1}{1 + \abs{G_0 (0) - \tilde{G}_0} + \abs{ \tilde{G}_0 - \tilde{F}_0 }}\right)  \left( \vphantom{\frac{1-|\gamma_0|^2}{|\gamma_0|^2}} 1-\abs{G_0 (0) - \tilde{G}_0} - \abs{ \tilde{G}_0 - \tilde{F}_0 } \right. \cr 
    &\left.-  \frac{9}{4}\sqrt{\frac{1}{\eta}} \lVert \epsilon^{(R)} \rVert  \right)
\end{align*}  
\end{lemma}

\begin{proof}
    The exact hit rate $\tilde{p}_0$ is
    \begin{align}
        \tilde{p}_0 = \left\Vert \frac{\gamma_0}{\sqrt{\mathcal{N}}}\sqrt{g_0}  + \epsilon^{(\rm polut)}\right\Vert^2
    \end{align}

An upper bound on $\tilde{p}_0$ can be written the following way where a maximum "destructive" interference is assumed for the vectors contributing to $\tilde{p}_0$
    \begin{align}
        \tilde{p}_0 &\geq \left(\frac{1}{\max_{\tilde\mu}\mathcal{N}}\right) \abs{   \gamma_0}^2 \left\vert \vphantom{\frac{1-|\gamma_0|^2}{|\gamma_0|^2}} \lVert  \sqrt{g_0}   \rVert^2  \right.\cr 
        &\left.-  2\sqrt{\frac{1-|\gamma_0|^2}{|\gamma_0|^2}}\lVert \sqrt{g_0} \rVert \lVert \epsilon^{(R)} \rVert \right.\cr 
        &\left.+ 2\frac{1-|\gamma_0|^2}{|\gamma_0|^2}\lVert \epsilon^{(R)} \rVert^2  \right\vert.
    \end{align}
Here, we have ignored the contributions from the terms proportional to $\| \epsilon^{(R)} \|^2$. How

    \begin{align}
        \tilde{p}_0 &\geq \left(\frac{1}{\max_{\tilde\mu}\mathcal{N}}\right) \abs{\gamma_0}^2 \left\vert \vphantom{\frac{1-|\gamma_0|^2}{|\gamma_0|^2}} \lVert  \sqrt{g_0}   \rVert^2 \right. \cr 
        &\left.-  2\sqrt{\frac{1}{\eta}}\lVert \sqrt{g_0} \rVert \lVert \epsilon^{(R)} \rVert \right\vert
    \end{align}

It will be useful to bound $\|\sqrt{g_0}\|^2$ in terms of $|G_0(0)-\tilde{G}_0|$ and $\sqrt{\abs{ \tilde{G}_0 - \tilde{F}_0 }}$. First, let us note that because
    \begin{align}
        \abs{G_0 (0) - \tilde{G}_0} = \abs{1 - \tilde{G}_0},
    \end{align}
we know that
\begin{align}
            \abs{1- \tilde{F_0}} \leq  \abs{G_0 (0) - \tilde{G}_0} + \abs{ \tilde{G}_0 - \tilde{F}_0 }.
\end{align}
Finally, because $\lVert \sqrt{g_0} \rVert^2 = \tilde{F}_0$, we obtain:
    \begin{align}
        1-\abs{G_0 (0) - \tilde{G}_0} - \abs{ \tilde{G}_0 - \tilde{F}_0 } \leq \lVert \sqrt{g_0} \rVert^2 \leq 1 + \abs{G_0 (0) - \tilde{G}_0} + \abs{ \tilde{G}_0 - \tilde{F}_0 }
    \end{align}

    \begin{align}
        \tilde{p}_0 &\geq \left(\frac{1}{\max_{\tilde\mu}\mathcal{N}}\right) \abs{\gamma_0}^2 \left( \vphantom{\frac{1-|\gamma_0|^2}{|\gamma_0|^2}} 1-\abs{G_0 (0) - \tilde{G}_0} - \abs{ \tilde{G}_0 - \tilde{F}_0 } \right. \cr 
        &\left.-  2\sqrt{\frac{1}{\eta}}\left(1 + \abs{G_0 (0) - \tilde{G}_0}/2 + \abs{ \tilde{G}_0 - \tilde{F}_0 }/2\right) \lVert \epsilon^{(R)} \rVert  \right)
    \end{align}
If we assume $\abs{G_0 (0) - \tilde{G}_0}\leq 1/8$ and $\abs{ \tilde{G}_0 - \tilde{F}_0 } \leq 1/8$ we can write
    \begin{align}
        \tilde{p}_0 &\geq \left(\frac{1}{\max_{\tilde\mu}\mathcal{N}}\right) \abs{\gamma_0}^2 \left( \vphantom{\frac{1-|\gamma_0|^2}{|\gamma_0|^2}} 1-\abs{G_0 (0) - \tilde{G}_0} - \abs{ \tilde{G}_0 - \tilde{F}_0 } \right. \cr 
        &\left.-  \frac{9}{4}\sqrt{\frac{1}{\eta}} \lVert \epsilon^{(R)} \rVert  \right)
    \end{align}
Now, having $\sqrt{\frac{1}{\eta}} \| \epsilon^{(R)} \| \leq 1/8$ is sufficient condition to ensure the positivity of the lower bound on $\tilde{p}_0$. Finally, we use $\max_{\tilde\mu}\mathcal{N} \leq 1 + \abs{G_0 (0) - \tilde{G}_0} + \abs{ \tilde{G}_0 - \tilde{F}_0 }$ from \Cref{lem:norm_error} to obtain
\begin{align}
    \tilde{p}_0 &\geq \abs{\gamma_0}^2\left(\frac{1}{1 + \abs{G_0 (0) - \tilde{G}_0} + \abs{ \tilde{G}_0 - \tilde{F}_0 }}\right)  \left( \vphantom{\frac{1-|\gamma_0|^2}{|\gamma_0|^2}} 1-\abs{G_0 (0) - \tilde{G}_0} - \abs{ \tilde{G}_0 - \tilde{F}_0 } \right. \cr 
    &\left.-  \frac{9}{4}\sqrt{\frac{1}{\eta}} \lVert \epsilon^{(R)} \rVert  \right)
\end{align}
\end{proof}

\begin{corollary}\label{cor:hit_rate}
Since we have the constraints $\sqrt{\frac{1}{\eta}} \| \epsilon^{(R)} \| \leq 1/8$, $\sqrt{\abs{ \tilde{G}_0 - \tilde{F}_0 }}\leq 1/8$, and $\abs{G_0(0)-\tilde{G}_0}\leq 1/8$, we can simplify the bound the following way
    $$
        \tilde{p}_0 \geq \eta \left( \frac{3}{8} \right)
    $$
\end{corollary}

\subsection{Aliasing/Discretization errors}

We are now about upper-bounding the quantity $\abs{G_m(0)-\tilde{G}_m}$: useful in estimating discretization errors and also used when $m=0$ for evaluating some of the bounds previously derived.
\begin{lemma}[Aliasing]
$$
|G_m(0)-\tilde{G}_m|=\sum_{k\neq 0} |G_m(-k)|
$$

\end{lemma}
\begin{proof}
We must also estimate the discretization errors coming from the coarseness of the GPE. We start by writing out the discrete version of moment estimation. Let us define
$$
 g_m(x)=\frac{1}{\sigma \sqrt{2\pi}}x^m e^{-(\frac{x-\mu}{\sigma})^2/2}
$$
The moment estimation sum can be written down in terms of an integral with a $\comb$ function
\begin{align}
\tilde{G}_m = \sum_{x=-\infty}^{\infty} g_m(x) &= \int \mathrm{d}x \comb(x)  g_m(x) \cr
    &= \int \mathrm{d}x \mathcal{F}^{-1}\left(\mathcal{F}(\comb) \star G_m\right)(x) \cr 
    &= \int \mathrm{d}x \mathcal{F}^{-1}\left(\comb \star G_m\right)(x) \cr 
    &= \int \mathrm{d}x \int \mathrm{d}\xi e^{i2\pi \xi x} \left(\sum_k G_m(\xi-k)\right) \cr
    &= \int \mathrm{d}\xi \delta(\xi) \left( \sum_k G_m(\xi-k)\right) \cr
    &=  \sum_k G_m(-k)
\end{align}
The continuous momenta are just $G_m(0)$, and for the raw momenta we do the replacement $\mu \to 0$. The discretization error for the different momenta is
$$
|G_m(0)-\tilde{G}_m|=\sum_{k\neq 0} |G_m(-k)|
$$
\end{proof}

To obtain a bound on discretization errors, we would like to first bound each $\abs{G_m(k)}$ and for that we will be using its moment generating properties. For that, we will instead bound $\abs{D^m_k G_0(k)}$. First, we will introduce the following lemma:
\begin{lemma}[Derivative bound]\label{lem:fnbound}
If the function $f$ is analytic on the complex plane with $|z|\leq r$, then 
\begin{align}
|f^{(n)}(z)| \leq \frac{M n! 2^n}{r^n},
\end{align}
for $|z| \leq \frac{r}{2}$, where $M=\max_{|z|\leq r} |f(z)|$.
\end{lemma}
\begin{proof}
$$
f^{(n)}(z)=\frac{n!}{2\pi i}\int_{|z|=r}\frac{f(\xi)}{(\xi-z)^{n+1}}d\xi
$$
Now, consider another disk $\gamma$ centered at z:
$$
|\xi - z|=\frac{r-|z|}{\omega}
$$
where $\omega \geq 1$ such that $\gamma$ is always contained within $|z|\leq r$. Thus, with $M = \max_{|z|\leq r} |f(z)|$,
\begin{align}
\bigg|f^{(n)}(z)\bigg|&=\bigg|\frac{n!}{2\pi i}\int_{|z|=r}\frac{f(\xi)}{(\xi-z)^{n+1}}d\xi\bigg|
\cr
&=\bigg|\frac{n!}{2\pi i}\int_{C}\frac{f(\xi)}{(\xi-z)^{n+1}}d\xi\bigg|
\cr
&\leqslant\frac{n!}{2\pi}\int_{C}\frac{M}{|\xi-z|^{n+1}}|d\xi|
\cr
&=\frac{M n!}{2\pi}\int_{0}^{2\pi}\frac{1}{|\xi-z|^{n+1}}\frac{r-|z|}{\omega}d\theta
\cr
&=\frac{M n!}{2\pi}\int_{0}^{2\pi}\frac{\omega^n}{(r-|z|)^{n}}d\theta
\cr
&=\frac{M n!\omega^n}{(r-|z|)^{n}}
\cr
&\leq \frac{M n!2^n}{r^{n}}.
\end{align}
\end{proof}

\begin{lemma}[Discretization Errors]\label{lem:disc_errors}
Given that $\exp{-2 \pi^2 \sigma^2 } \leq 1/2$, we know that
$$
|\tilde{G}_m-G_m(0)|=2\sum^{\infty}_{k=1} |G_m(k)| \leq 4\frac{\exp{2 \pi \delta_1 \abs{\mu}} \exp{2\pi^2(\delta_1^2+2\delta_1) \sigma^2}\exp{-2 \pi^2 \sigma^2 }m!}{\pi^m\delta_1^m}.
$$
\end{lemma}
\begin{proof}
First, we not that
$$
G_0(k) = e^{-2 \pi  k \left(\pi  k \sigma ^2+i \mu \right)}.
$$
Thus, on the complex disc around $k$ with radius $\delta$, an upper bound on the magnitude of $G_0$ is
$$
M \leq \exp{2\pi \delta_1 \abs{\mu}} \exp{2 \pi^2 \delta_1^2 \sigma^2 } \exp{-2 \pi^2 k^2 \sigma^2 } \exp{4 \pi^2 k \delta_1 \sigma^2 }
$$
Thus,
$$
|G_m(k)| \leq \frac{\exp{2 \pi \delta_1 \abs{\mu}} \exp{2 \pi^2 \delta_1^2 \sigma^2 } \exp{-2 \pi^2 k^2 \sigma^2 } \exp{4 \pi^2 k \delta_1 \sigma^2 } m! (2 )^m} {(2\pi)^m\delta_1^m} 
$$

\begin{align*}
&|\tilde{G}_m-G_m(0)|=2\sum^{\infty}_{k=1} |G_m(k)| \cr 
&\leq 2\frac{\exp{2 \pi \delta_1 \abs{\mu}} \exp{2\pi^2\delta_1^2 \sigma^2}m!2^m}{(2\pi)^m\delta_1^m}\sum^{\infty}_{k=1} \exp{4 \pi^2 k \delta_1 \sigma^2 }\exp{-2 \pi^2 k^2 \sigma^2 }
\end{align*}
If we assume that 
$$
\exp{-2\pi^2 \sigma^2} \leq 1
$$
we can instead use the following bound
\begin{align*}
&|\tilde{G}_m-G_m(0)| \cr 
&\leq 2\frac{\exp{2 \pi \delta_1 \abs{\mu}} \exp{2\pi^2\delta_1^2 \sigma^2}m!2^m}{(2\pi)^m\delta_1^m} \sum^{\infty}_{k=1} \exp{4 \pi^2 k \delta_1 \sigma^2 }\exp{-2 \pi^2 k \sigma^2 }
\end{align*}
Moreover, if we assume
$$
\exp{-2 \pi^2 \sigma^2(1-2\delta_1) } \leq 1/2,
$$
\begin{align*}
&|\tilde{G}_m-G_m(0)| \cr
&\leq 2\frac{\exp{2 \pi \delta_1 \abs{\mu}} \exp{2\pi^2(\delta_1^2 +2 \delta_1 )\sigma^2}\exp{-2 \pi^2 \sigma^2 }m!2^m}{(2\pi)^m\delta_1^m}\cr   
&\quad\quad\times\left(\frac{1}{1-\exp{-2 \pi^2 \sigma^2(1-2\delta_1) }}\right) \cr
& \leq 4\frac{\exp{2 \pi \delta_1 \abs{\mu}} \exp{2\pi^2(\delta_1^2 +2 \delta_1 ) \sigma^2}\exp{-2 \pi^2 \sigma^2 }m!}{\pi^m\delta_1^m}.
\end{align*}
\end{proof}

\subsection{Truncation error and excited state contributions}

In order to estimate the truncation and contamination errors, we start by proving the following lemma:
\begin{lemma}\label{lem:HmFmError}
Given $f,h\in \mathbb{C}^{2K+1}$ with indices running through $x\in \{x_0-K,x_0-K+1,\dots,x_0+K\}$, and given the two definitions
\begin{align*}
\tilde{H}_m (k) &= \sum_{x} \exp{- 2 \pi i x k/(2K+1) } x^m |h_x|^2 \cr 
\tilde{F}_m(k) &= \sum_{x} \exp{-2\pi i x k/(2K+1)} x^{m} |f_x|^2,
\end{align*}
provided $\max\abs{x}\leq 2K$, the following inequalities hold:
\begin{align*}
|\tilde{H}_0(0)-\tilde{F}_0(0)| &\leq \||h|^2 - |f|^2\|_1 \cr 
|\tilde{H}_m(0)-\tilde{F}_m(0)| &\leq\frac{m!(2K+1)^m}{\pi^m\delta_2^m} \exp{2\pi \delta_2 }\||h|^2 - |f|^2\|_1  \cr 
&\leq \frac{m!(2K+1)^m}{\pi^m\delta_2^m} \exp{2\pi \delta_2 }\left(\|\epsilon\|^2 + 2\| h \| \| \epsilon \|\right),
\end{align*}
where $\delta_2\in \mathbb{R}^+$.

\end{lemma}

\begin{proof}
We take into consideration the following generating functions
\begin{align}
\tilde{H}_m(k) &= (-i)^m \left( \frac{2K+1}{2\pi}\right)^m D^m_k \tilde{H}_0(k) \cr
\tilde{F}_m(k) &= (-i)^m \left( \frac{2K+1}{2\pi}\right)^m D^m_k \tilde{F}_0(k).
\end{align}
Now, we are interested in bounding the magnitude of $\tilde{H}_m(0)-\tilde{F}_m(0)$ for which there is also a similar moment generating function by the linearity of the differentiation operation, and thus can use \Cref{lem:fnbound} around $k=0$. The maximum value $M$ from \Cref{lem:fnbound} would be:
$$
\max_{|k|\leq \delta_2}|\tilde{H}_0(k)-\tilde{F}_0(k)| \leq \exp{2\pi \delta_2 }\||h|^2-|f|^2\|_1,
$$
and with this
$$
|\tilde{H}_m(0)-\tilde{F}_m(0)| \leq \frac{m!(2K+1)^m}{\pi^m\delta_2^m} \exp{2\pi \delta_2 }\||h|^2-|f|^2\|_1.
$$

We also have that
\begin{align}
\||h|^2 - |f|^2\|_1 &\leq \|\epsilon\|_2^2 + \| h^* \epsilon + h\epsilon^{*} \|_1 \cr 
&\leq \|\epsilon\|_2^2 + \| h^* \epsilon \|_1 +  \|h\epsilon^{*} \|_1 \cr 
&\leq \|\epsilon\|_2^2 + 2\| h \|_2 \| \epsilon \|_2.
\end{align}
Here, $f = h - \epsilon$. Thus, we finally get
$$
|\tilde{H}_m(0)-\tilde{F}_m(0)| \leq \frac{m!(2K+1)^m}{\pi^m\delta_2^m} \exp{2\pi \delta_2 }\left(\|\epsilon\|^2 + 2\| h \| \| \epsilon \|\right).
$$

\end{proof}

\begin{lemma}
Given the following definitions
\begin{align*}
\tilde{G}_m (k) &= \sum^{\infty}_{x=-\infty} \exp{- 2 \pi i x k/(2K+1) } x^m g_0(x) \cr 
\tilde{F}^{(\rm polut)}_m(k) &= \sum_{x} \exp{-2\pi i x k/(2K+1)} x^{m} |f^{(\rm polut)}_x|^2,
\end{align*}
where $f\in \mathbb{C}^{2K+1}$ with indices running through $x\in \{x_0-K,x_0-K+1,\dots,x_0+K\}$ such that
\begin{align*}
    \left(f\right)_x &= \sqrt{g_0(x)}, \cr 
    \left(f^{(\rm polut)}\right)_x &= \left(f\right)_x + \frac{\sqrt{\mathcal{N}}}{\gamma_0}\left(\epsilon^{(\rm polut)}\right)_x,    
\end{align*}
and $\max\abs{x}\leq 2K$, we find the error between the two is bounded through
\begin{align*}
    |\tilde{G}_m(0)-\tilde{F}^{(\rm polut)}_m(0)| &\leq \frac{m!(2K+1)^m}{\pi^m\delta_2^m} \exp{2\pi \delta_2 }\left\Vert\tilde{g}-\abs{ f^{(\rm polut)} }^2\,\,\right\Vert_1 \cr 
    &\leq\frac{11}{4} \frac{m!(2K+1)^m}{\pi^m\delta_2^m} \exp{2\pi \delta_2 } \left( \vphantom{\sqrt{\frac{1}{1}}} \sqrt{\abs{ \tilde{G}_0 - \tilde{F}_0 }}  + \sqrt{\frac{5}{3}}\sqrt{\frac{1}{\eta}}\|\epsilon^{(R)}\| \right).
\end{align*}
where $\delta_2\in \mathbb{R}^+$.
\end{lemma}

\begin{proof}
First of all, we note that $\tilde{G}_0(k)$ can be written the following way
\begin{align}
    \tilde{G}_0 (k) &= \sum^{\infty}_{x=-\infty} \exp{- 2 \pi i x k/(2K+1) } g_0(x) \cr
                    &=\sum^{x_0+K}_{x=x_0-K} \exp{- 2 \pi i x k/(2K+1) } \tilde{g}(x), \cr
\end{align}
where 
$$
\tilde{g}(x) = \sum_j g_0(x-j(2K+1)).
$$
Thus, with $\tilde{G_0}$ written in terms of $\tilde{g}$ and with the results from \Cref{lem:HmFmError} we obtain that the error $|\tilde{G}_m(0)-\tilde{F}^{(\rm polut)}_m(0)|$ is
\begin{align}\label{eq:GmFmError}
|\tilde{G}_m(0)-\tilde{F}^{(\rm polut)}_m(0)| &\leq \frac{m!(2K+1)^m}{\pi^m\delta_2^m} \exp{2\pi \delta_2 }\left\Vert\tilde{g}-\abs{ f^{(\rm polut)} }^2\,\,\right\Vert_1 \cr 
&\leq \frac{m!(2K+1)^m}{\pi^m\delta_2^m} \exp{2\pi \delta_2 }\cr
&\quad\quad\times\left( \left\Vert\sqrt{\tilde{g}}- f^{(\rm polut)}\right\Vert^2 + 2 \left\Vert \sqrt{\tilde{g}} \right\Vert \left\Vert\sqrt{\tilde{g}}- f^{(\rm polut)}\,\,\right\Vert  \right), \cr
\end{align}
where $\delta_2\in \mathbb{R}^+$. We can obtain the following bound for $\left\Vert \sqrt{\tilde{g}} - f^{(\rm polut)} \right\Vert$:

\begin{align}
\left\Vert \sqrt{\tilde{g}} - f^{(\rm polut)} \right\Vert &= \left\Vert \sqrt{\tilde{g}} - \sqrt{g}  -  \epsilon^{(\rm polut)} \right\Vert \cr
&\leq \left\Vert \sqrt{\tilde{g}} - \sqrt{g} \right\Vert + \frac{\sqrt{\mathcal{N}}}{|\gamma_0|}\left\Vert   \epsilon^{(\rm polut)} \right\Vert.
\end{align}
Now, we would like to show that
\begin{align}
    \left\Vert \sqrt{\tilde{g}} - \sqrt{g} \right\Vert \leq \sqrt{\abs{ \tilde{G}_0 - \tilde{F}_0 }}. \cr 
\end{align}
For that, we first note that 
\begin{align}
    \Tilde{g}_x \geq g_x \geq 0. \cr 
\end{align}
Then, it follows that
\begin{align}
    (\sqrt{\tilde{g}}_x - \sqrt{g}_x)^2  \leq \tilde{g}_x - g_x.
\end{align}
We can understand this more easily if we write $\tilde{g}_x = c_x g_x$, where $c_x\leq 1$
\begin{align}
    (\sqrt{\tilde{g}}_x - \sqrt{c_x}\sqrt{\Tilde{g}}_x)^2 = \tilde{g}_x \left(1-\sqrt{c_x}\right)^2  &\leq \tilde{g}_x - g_x = \Tilde{g}_x(1-c_x) \cr 
    \left(1-\sqrt{c_x}\right)^2 &\leq 1-c_x.
\end{align}
Finally, we sum over all $x$ to obtain
\begin{align}
    \left\Vert \sqrt{\tilde{g}} - \sqrt{g} \right\Vert^2 \leq \left\Vert \tilde{g} - g \right\Vert_1 = \abs{ \tilde{G}_0 - \tilde{F}_0 } . \cr 
\end{align}

With this,
\begin{align}
    &\left( \left\Vert\sqrt{\tilde{g}}- f^{(\rm polut)}\right\Vert^2 + 2 \left\Vert \sqrt{\tilde{g}} \right\Vert \left\Vert\sqrt{\tilde{g}}- f^{(\rm polut)}\,\,\right\Vert  \right) \cr 
    &\leq \left( \sqrt{\abs{ \tilde{G}_0 - \tilde{F}_0 }} + \frac{\sqrt{\mathcal{N}}}{|\gamma_0|}\|\epsilon^{(\rm polut)}\| \right)^2 + 2 \|\sqrt{\tilde{g}}\| \left( \sqrt{\abs{ \tilde{G}_0 - \tilde{F}_0 }} + \frac{\sqrt{\mathcal{N}}}{|\gamma_0|}\|\epsilon^{(\rm polut)}\| \right) \cr 
\end{align}
Now, we note that there exist the following bound for $\left\Vert \sqrt{\tilde{g}} \right\Vert^2$
\begin{align}
    \left\Vert \sqrt{\tilde{g}} \right\Vert^2 & \leq \| \tilde{g} - g_0 \|_1 + \| g_0 \|_1 \cr
    &\leq  1 + \abs{G_0 (0)- \tilde{G}_0} + 2\abs{ \tilde{G}_0 - \tilde{F}_0 }
\end{align}
With this, the term in parentheses in \Cref{eq:GmFmError} becomes
\begin{align}
    &\left( \left\Vert\sqrt{\tilde{g}}- f^{(\rm polut)}\right\Vert^2 + 2 \left\Vert \sqrt{\tilde{g}} \right\Vert \left\Vert\sqrt{\tilde{g}}- f^{(\rm polut)}\,\,\right\Vert  \right) \cr 
    &\leq \left( \sqrt{\abs{ \tilde{G}_0 - \tilde{F}_0 }} + \frac{\sqrt{\mathcal{N}}}{|\gamma_0|}\|\epsilon^{(\rm polut)}\| \right)^2 \cr 
    &+ 2 \left(1 + \abs{G_0 (0)- \tilde{G}_0}/2 + \abs{ \tilde{G}_0 - \tilde{F}_0 }\right) \left( \sqrt{\abs{ \tilde{G}_0 - \tilde{F}_0 }} + \frac{\sqrt{\mathcal{N}}}{|\gamma_0|}\|\epsilon^{(\rm polut)}\| \right). \cr 
\end{align}
Finally, through the results of \Cref{lem:contamination_bounds}, \Cref{lem:norm_error}, and the fact that $\eta \leq \abs{\gamma_0}^2$ and $1>1-\eta$, we obtain
\begin{align}
    |\tilde{G}_m(0)-\tilde{F}^{(\rm polut)}_m(0)| &\leq \frac{m!(2K+1)^m}{\pi^m\delta_2^m} \exp{2\pi \delta_2 }\left\Vert\tilde{g}-\abs{ f^{(\rm polut)} }^2\,\,\right\Vert_1 \cr 
    &\leq \frac{m!(2K+1)^m}{\pi^m\delta_2^m} \exp{2\pi \delta_2 }\cr
    &\quad\times\left(\left( \vphantom{\sqrt{\frac{1}{1}}} \sqrt{\abs{ \tilde{G}_0 - \tilde{F}_0 }}  \right.\right.\cr
    &\left.\left.+ \sqrt{\frac{1+\abs{G_0(0)-\tilde{G}_0}+\abs{ \tilde{G}_0 - \tilde{F}_0 }}{1-\abs{G_0(0)-\tilde{G}_0}-\abs{ \tilde{G}_0 - \tilde{F}_0 }}}\sqrt{\frac{1}{\eta}}\|\epsilon^{(R)}\| \right)^2 \right.\cr 
    &\left.+ 2 \left(1 + \abs{G_0 (0)- \tilde{G}_0}/2 + \abs{ \tilde{G}_0 - \tilde{F}_0 }\right) \left( \vphantom{\frac{\sqrt{1}}{\sqrt{1}}} \sqrt{\abs{ \tilde{G}_0 - \tilde{F}_0 }} \right.\right.\cr 
    &\left.\left.+ \sqrt{\frac{1+\abs{G_0(0)-\tilde{G}_0}+\abs{ \tilde{G}_0 - \tilde{F}_0 }}{1-\abs{G_0(0)-\tilde{G}_0}-\abs{ \tilde{G}_0 - \tilde{F}_0 }}}\sqrt{\frac{1}{\eta}}\|\epsilon^{(R)}\| \right)\right). \cr   
\end{align}
If we assume $\abs{G_0(0)-\Tilde{G}_0}\leq 1/8$ and $\abs{ \tilde{G}_0 - \tilde{F}_0 } \leq 1/8$, as well as $\sqrt{\abs{ \tilde{G}_0 - \tilde{F}_0 }} \leq 1/\sqrt{8} \leq 1/2$, we can further simplify

\begin{align}
    |\tilde{G}_m(0)-\tilde{F}^{(\rm polut)}_m(0)| &\leq \frac{m!(2K+1)^m}{\pi^m\delta_2^m} \exp{2\pi \delta_2 }\left\Vert\tilde{g}-\abs{ f^{(\rm polut)} }^2\,\,\right\Vert_1 \cr 
    &\leq \frac{m!(2K+1)^m}{\pi^m\delta_2^m} \exp{2\pi \delta_2 }\cr
    &\quad\quad\times\left(\left(  1/2 + \sqrt{\frac{5}{3}}1/8 \right)\left( \vphantom{\sqrt{\frac{1}{1}}} \sqrt{\abs{ \tilde{G}_0 - \tilde{F}_0 }}  + \sqrt{\frac{5}{3}}\sqrt{\frac{1}{\eta}}\|\epsilon^{(R)}\| \right) \right.\cr 
    &\left.+  \frac{19}{8} \left( \vphantom{\frac{\sqrt{1}}{\sqrt{1}}} \sqrt{\abs{ \tilde{G}_0 - \tilde{F}_0 }} + \sqrt{\frac{5}{3}}\sqrt{\frac{1}{\eta}}\|\epsilon^{(R)}\| \right)\right). \cr
    &\leq \frac{m!(2K+1)^m}{\pi^m\delta_2^m} \exp{2\pi \delta_2 }\left\Vert\tilde{g}-\abs{ f^{(\rm polut)} }^2\,\,\right\Vert_1 \cr 
    &\leq\frac{25}{8} \frac{m!(2K+1)^m}{\pi^m\delta_2^m} \exp{2\pi \delta_2 } \left( \vphantom{\sqrt{\frac{1}{1}}} \sqrt{\abs{ \tilde{G}_0 - \tilde{F}_0 }}  + \sqrt{\frac{5}{3}}\sqrt{\frac{1}{\eta}}\|\epsilon^{(R)}\| \right). \cr 
\end{align}

\end{proof}

\subsection{Bound on the error coming from renormalization, $\varepsilon^{\rm norm}_m$}

\begin{lemma}
Assuming that $\abs{\mu} \leq K$, we can bound $\varepsilon^{\rm norm}_m$ through
    \begin{align*}
        \varepsilon^{\rm norm}_m &\leq \left(\frac{128}{45}\right)\left(\frac{ \exp{2 \pi \delta_3 \abs{\mu}} \exp{2 \pi^2 \delta_3^2 \sigma^2 }  m!}{\pi^m \delta_3^m}\right)\cr 
        &\quad\times\left(  3\abs{G_0 (0) - \tilde{G}_0} + 3\abs{ \tilde{G}_0 - \tilde{F}_0 } + \frac{9}{4}\sqrt{\frac{1}{\eta}} \lVert \epsilon^{(R)} \rVert  \right)
    \end{align*}
\end{lemma}

\begin{proof}
    First, we recall that
    \begin{align}
        \varepsilon^{\rm norm}_m = \left| G_m(0) - \frac{1}{\mathcal{N}}\left(\frac{\abs{\gamma_0}^2}{\tilde{p}_0}\right) G_m(0) \right| = \abs{ 1 - \frac{1}{\mathcal{N}}\left(\frac{\abs{\gamma_0}^2}{\tilde{p}_0}\right) } \left| G_m(0) \right|.
    \end{align}
    Now, assuming that $\abs{\mu} \leq K$, a bound for $\left| G_m(0) \right| $ is
    $$
    \left| G_m(0) \right| = \left(\frac1{2\pi}\right)^m \abs{D^m_k G_0(k)}\Big\vert_{k=0} \leq  \frac{ \exp{2 \pi \delta_3 \abs{\mu}} \exp{2 \pi^2 \delta_3^2 \sigma^2 }  m!}{\pi^m \delta_3^m}.
    $$
    Moreover,
    \begin{align}
    &\frac{1}{\mathcal{N}}\left(\frac{\abs{\gamma_0}^2}{\tilde{p}_0}\right) - 1 \leq \cr 
    &\left(\frac{1}{ 1-\abs{ \tilde{G}_0 - \tilde{F}_0 }-|G_0(0)-\tilde{G}_0|}\right)\cr 
    &\qquad\times \left(\frac{1}{\left(\frac{1}{1 + \abs{G_0 (0) - \tilde{G}_0} + \abs{ \tilde{G}_0 - \tilde{F}_0 }}\right)  \left( \vphantom{\frac{1-|\gamma_0|^2}{|\gamma_0|^2}} 1-\abs{G_0 (0) - \tilde{G}_0} - \abs{ \tilde{G}_0 - \tilde{F}_0 } -  \frac{9}{4}\sqrt{\frac{1}{\eta}} \lVert \epsilon^{(R)} \rVert  \right)}\right) - 1 \cr 
    &\leq \frac{ 3\abs{G_0 (0) - \tilde{G}_0} + 3\abs{ \tilde{G}_0 - \tilde{F}_0 } + \frac{9}{4}\sqrt{\frac{1}{\eta}} \lVert \epsilon^{(R)} \rVert  }{\left( 1-\abs{ \tilde{G}_0 - \tilde{F}_0 }-|G_0(0)-\tilde{G}_0|\right) \left( 1-\abs{G_0 (0) - \tilde{G}_0} - \abs{ \tilde{G}_0 - \tilde{F}_0 } -  \frac{9}{4}\sqrt{\frac{1}{\eta}} \lVert \epsilon^{(R)} \rVert  \right)} \cr 
    &\leq \frac{ 3\abs{G_0 (0) - \tilde{G}_0} + 3\abs{ \tilde{G}_0 - \tilde{F}_0 } + \frac{9}{4}\sqrt{\frac{1}{\eta}} \lVert \epsilon^{(R)} \rVert  }{\left( 3/4 \right) \left( 32/32-4/32 - 4/32 -  \frac{9}{32}  \right)} \cr 
    &\leq \left(\frac{128}{45}\right)\left(  3\abs{G_0 (0) - \tilde{G}_0} + 3\abs{ \tilde{G}_0 - \tilde{F}_0 } + \frac{9}{4}\sqrt{\frac{1}{\eta}} \lVert \epsilon^{(R)} \rVert  \right).
    \end{align}
\end{proof}

\subsection{Requirements on Time evolution lenght, $T=2^q$}

\begin{theorem}\label{thm:2q}
Provided that $\abs{G_0(0)-\Tilde{G}_0} \leq 1/8$, $\abs{ \tilde{G}_0 - \tilde{F}_0 } \leq 1/8$, $\| \epsilon^{(R)} \|/\sqrt{\eta} \leq 1/8$, $1/2^q\leq \Delta/3$, and we remain in
    \begin{align*}
        \tilde\sigma  \leq \frac{1}{2^{1/4}}\sqrt{\frac{1}{2^q}}\sqrt{\frac{\Delta}{12\pi}},
    \end{align*}
for an error $\tilde\varepsilon_m$, we need
\begin{align*}
     2^q& \geq \frac{\sqrt{m}}{2\pi \tilde\sigma}\sqrt{1+3\log{\frac{m}{4 e \pi^2 \tilde\sigma^2 }\left(\frac{(m!)^2 C(\eta)}{\tilde\varepsilon_m}\right)^{2/m}}}\cr  
     &\quad\quad \text{for }\frac{m}{2 e \pi^2 \tilde\sigma^2 }\left(\frac{(m!)^2 C(\eta)}{\tilde\varepsilon_m}\right)^{2/m} > 1.    \cr
\end{align*}
where
\begin{align*}
    C(\eta) &= \left[\left(\frac{128}{45}\right)\exp{12}  \left(12 + 3 + \frac{9}{4}\sqrt{\frac{1}{\eta}}   \right) \right.\cr
    &\left. +10\exp{12} + \frac{55}{8} \exp{2} \left(  1  + \sqrt{\frac{5}{3}}\sqrt{\frac{1}{\eta}} \right)\right].
\end{align*}

\end{theorem}

\begin{proof}
    We group all the error terms
    \begin{align}
\varepsilon^{\rm norm}_m &
    \leq \left(\frac{128}{45}\right)\left(\frac{ \exp{2 \pi \delta_3 \abs{\mu}} \exp{2 \pi^2 \delta_3^2 \sigma^2 }  m!}{\pi^m \delta_3^m}\right) \cr 
    &\times\left(  12\frac{\exp{2 \pi \delta_1 \abs{\mu}} \exp{2\pi^2(\delta_1^2 + 2 \delta_1)\sigma^2}m!}{\pi^m\delta_1^m}\exp{-2 \pi^2 \sigma^2 } \right.\cr  
    &\left.+ 3\abs{ \tilde{G}_0 - \tilde{F}_0 } + \frac{9}{4}\sqrt{\frac{1}{\eta}} \lVert \epsilon^{(R)} \rVert  \right) \cr
\varepsilon^{\rm discret}_m &
    \leq \frac{1}{\mathcal{N}}\left(\frac{\abs{\gamma_0}^2}{\tilde{p}_0}\right)4\frac{\exp{2 \pi \delta_1 \abs{\mu}} \exp{2\pi^2(\delta_1^2 + 2 \delta_1)\sigma^2}m!}{\pi^m\delta_1^m}\exp{-2 \pi^2 \sigma^2 } \cr 
\varepsilon^{\rm trunc+cont}_{m} &
    \leq \frac{1}{\mathcal{N}}\left(\frac{\abs{\gamma_0}^2}{\tilde{p}_0}\right)\frac{11}{4} \frac{m!(2K+1)^m}{\pi^m\delta_2^m} \exp{2\pi \delta_2 } \left( \vphantom{\sqrt{\frac{1}{1}}} \sqrt{\abs{ \tilde{G}_0 - \tilde{F}_0 }}  + \sqrt{\frac{5}{3}}\sqrt{\frac{1}{\eta}}\|\epsilon^{(R)}\| \right). \cr 
\end{align}

We now note the prefactor $\frac{1}{\mathcal{N}}\frac{\abs{\gamma_0}^2}{\Tilde{p}_0}$ can be bounded through the constraints we have imposed: $\abs{G_0(0)-\Tilde{G}_0} \leq 1/8$, $\abs{ \tilde{G}_0 - \tilde{F}_0 } \leq 1/8$, and $\| \epsilon^{(R)} \|/\sqrt{\eta} \leq 1/8$
    $$
    \frac{1}{\mathcal{N}}\frac{\abs{\gamma_0}^2}{\Tilde{p}_0} \leq \left(\frac{4}{3}\right)\left(\frac{4}{5}\right)\left(\frac{32}{15}\right) = \frac{512}{225} \leq 5/2
    $$
With this, the bounds become:
    \begin{align}
\varepsilon^{\rm norm}_m &
    \leq \left(\frac{128}{45}\right)\left(\frac{ \exp{2 \pi \delta_3 \abs{\mu}} \exp{2 \pi^2 \delta_3^2 \sigma^2 }  m!}{\pi^m \delta_3^m}\right) \cr 
    & \times \left(12\frac{\exp{2 \pi \delta_1 \abs{\mu}} \exp{2\pi^2(\delta_1^2 + 2 \delta_1)\sigma^2}m!}{\pi^m\delta_1^m}\exp{-2 \pi^2 \sigma^2 } \right.\cr  
    &\left.\quad + 3\abs{ \tilde{G}_0 - \tilde{F}_0 } + \frac{9}{4}\sqrt{\frac{1}{\eta}} \lVert \epsilon^{(R)} \rVert  \right) \cr
\varepsilon^{\rm discret}_m &
    \leq 10\frac{\exp{2 \pi \delta_1 \abs{\mu}} \exp{2\pi^2(\delta_1^2 + 2 \delta_1)\sigma^2}m!}{\pi^m\delta_1^m}\exp{-2 \pi^2 \sigma^2 } \cr 
\varepsilon^{\rm trunc+cont}_{m} &
    \leq \frac{55}{8} \frac{m!(2K+1)^m}{\pi^m\delta_2^m} \exp{2\pi \delta_2 } \left( \vphantom{\sqrt{\frac{1}{1}}} \sqrt{\abs{ \tilde{G}_0 - \tilde{F}_0 }}  + \sqrt{\frac{5}{3}}\sqrt{\frac{1}{\eta}}\|\epsilon^{(R)}\| \right)
\end{align}
Given that $\delta_1\leq 1$ and $\delta_3\leq 1$, we observe that
\begin{align}
    \frac{ \exp{2 \pi \delta_3 \abs{\mu}} \exp{2 \pi^2 \delta_3^2 \sigma^2 }  m!}{\pi^m \delta_3^m} \geq 1 \cr
    \frac{ \exp{2 \pi \delta_1 \abs{\mu}} \exp{2\pi^2(\delta_1^2 + 2 \delta_1)\sigma^2}  m!}{\pi^m \delta_1^m} \geq 1.
\end{align}
Thus,
    \begin{align}
\varepsilon^{\rm norm}_m &
    \leq \left(\frac{128}{45}\right)\left(\frac{ \exp{2 \pi \delta_3 \abs{\mu}} \exp{2 \pi^2 \delta_3^2 \sigma^2 }  m!}{\pi^m \delta_3^m}\right) \cr 
    &\times\left(\frac{ \exp{2 \pi \delta_1 \abs{\mu}} \exp{2\pi^2(\delta_1^2 + 2 \delta_1)\sigma^2}  m!}{\pi^m \delta_1^m}\right) \cr 
    & \times \left(12\exp{-2 \pi^2 \sigma^2 } + 3\abs{ \tilde{G}_0 - \tilde{F}_0 } + \frac{9}{4}\sqrt{\frac{1}{\eta}} \lVert \epsilon^{(R)} \rVert  \right) \cr
\varepsilon^{\rm discret}_m &
    \leq 10\left(\frac{\exp{2 \pi \delta_1 \abs{\mu}} \exp{2\pi^2(\delta_1^2 + 2 \delta_1)\sigma^2}m!}{\pi^m\delta_1^m}\right)\cr 
    &\times\left(\frac{ \exp{2 \pi \delta_3 \abs{\mu}} \exp{2 \pi^2 \delta_3^2 \sigma^2 }  m!}{\pi^m \delta_3^m}\right)\exp{-2 \pi^2 \sigma^2 } \cr 
\end{align}
Now, we set $\delta_1=\delta_3=\frac{1}{\pi2^q}$
    \begin{align}
\varepsilon^{\rm norm}_m &
    \leq \left(\frac{128}{45}\right)(2^q)^{2m}\exp{4 \frac{ \abs{\mu}}{2^q} } \exp{4\frac{\sigma^2}{(2^q)^2}}\exp{4\pi\frac{\sigma^2}{2^q}}\left(m!\right)^2 \cr 
    & \times \left(12\exp{-2 \pi^2 \sigma^2 } + 3\abs{ \tilde{G}_0 - \tilde{F}_0 } + \frac{9}{4}\sqrt{\frac{1}{\eta}} \lVert \epsilon^{(R)} \rVert  \right) \cr
\varepsilon^{\rm discret}_m &
    \leq 10(2^q)^{2m}\exp{4 \frac{ \abs{\mu}}{2^q} } \exp{4\frac{\sigma^2}{(2^q)^2}}\exp{4\pi\frac{\sigma^2}{2^q}}\left(m!\right)^2\cr 
    &\times\exp{-2 \pi^2 \sigma^2 }. \cr 
\end{align}
Moreover, we have $\abs{\mu}\leq 2^q$, and assume $\sigma\leq \sqrt{2^q/\pi}$ which gets us
\begin{align}
\varepsilon^{\rm norm}_m &
    \leq \left(\frac{128}{45}\right)(2^q)^{2m}\exp{12}\left(m!\right)^2 \cr 
    & \times \left(12\exp{-2 \pi^2 \sigma^2 } + 3\abs{ \tilde{G}_0 - \tilde{F}_0 } + \frac{9}{4}\sqrt{\frac{1}{\eta}} \lVert \epsilon^{(R)} \rVert  \right) \cr
\varepsilon^{\rm discret}_m &
    \leq 10(2^q)^{2m}\exp{12} \left(m!\right)^2\cr 
    &\times\exp{-2 \pi^2 \sigma^2 }  
\end{align}
Next, we set $\delta_2 = 1/\pi$ and use $m!(2^q)^m\leq (m!(2^q)^m)^2 $
\begin{align}
\varepsilon^{\rm trunc+cont}_{m} &
    \leq \frac{55}{8} (m!)^2(2^q)^{2m}\exp{2} \left( \vphantom{\sqrt{\frac{1}{1}}} \sqrt{\abs{ \tilde{G}_0 - \tilde{F}_0 }}  + \sqrt{\frac{5}{3}}\sqrt{\frac{1}{\eta}}\|\epsilon^{(R)}\| \right).    
\end{align}
We can now put the errors together and factorize some common terms
\begin{align}
\tilde\varepsilon_m & \leq (2^q)^{2m} (m!)^2 \cr 
    &\times\left[\left(\frac{128}{45}\right)\exp{12} \right.\cr 
    &\left. \times \left(12\exp{-2 \pi^2 \sigma^2 } + 3\abs{ \tilde{G}_0 - \tilde{F}_0 } + \frac{9}{4}\sqrt{\frac{1}{\eta}} \lVert \epsilon^{(R)} \rVert  \right) \right.\cr
    &\left. +10\exp{12}\exp{-2 \pi^2 \sigma^2 } \right.\cr
    &\left. + \frac{55}{8} \exp{2} \left( \vphantom{\sqrt{\frac{1}{1}}} \sqrt{\abs{ \tilde{G}_0 - \tilde{F}_0 }}  + \sqrt{\frac{5}{3}}\sqrt{\frac{1}{\eta}}\|\epsilon^{(R)}\| \right)\right].    
\end{align}

    We note that the exponent in $\exp{-2\pi^2\sigma^2}$ and the implied exponents from $\abs{ \tilde{G}_0 - \tilde{F}_0 }$,$\|\epsilon^{(R)}\|^2$ and $\sqrt{\abs{ \tilde{G}_0 - \tilde{F}_0 }}$ from our previous bounds are all different. This complicates the solution of the inequality.  For example the exponentials coming from the bounds on discretization errors and truncation errors ($\sqrt{\abs{ \tilde{G}_0 - \tilde{F}_0 }}$, \Cref{lem:eps1_error}) are:
    \begin{align}
        \exp{-2 \pi^2 \sigma^2 }&=\exp{-2 \pi^2 \sigma^2 },\cr 
        \abs{ \tilde{G}_0 - \tilde{F}_0 }&\leq\exp{-\frac{(K-1/2)^2}{2\sigma^2}} \cr 
        \sqrt{\abs{ \tilde{G}_0 - \tilde{F}_0 }}&\leq\exp{-\frac{(K-1/2)^2}{4\sigma^2}} \cr
        \|\epsilon^{(R)}\|&\leq\exp{-\frac{(2^q\Delta -K-1/2)^2}{4 \sigma^2}}
    \end{align}
    respectively. First, we can bound both $\abs{ \tilde{G}_0 - \tilde{F}_0 }$ and $ \sqrt{\abs{ \tilde{G}_0 - \tilde{F}_0 }}$ with a common exponential:
    \begin{align}
        \exp{-2 \pi^2 \sigma^2 }&=\exp{-2 \pi^2 \sigma^2 },\cr 
        \abs{ \tilde{G}_0 - \tilde{F}_0 }\leq\sqrt{\abs{ \tilde{G}_0 - \tilde{F}_0 }}&\leq\exp{-\frac{(K-1/2)^2}{4\sigma^2}} \cr 
        \|\epsilon^{(R)}\|&\leq\exp{-\frac{(2^q\Delta -K-1/2)^2}{4 \sigma^2}}
    \end{align}    
    Now, we will update the bound on $\| \epsilon^{(R)} \|$. First, we impose the restriction
    \begin{align}
        2K + 1 = \lfloor (2/3) \Delta 2^q \rfloor,
\end{align}
which means
\begin{align}
            2K +1 \leq \Delta 2^q, \cr
        \Delta 2^q \geq 2K.
\end{align}
    Thus, a new, looser bound on $\|\epsilon^{(R)}\|$ (\Cref{lem:contamination_bounds}) is
    \begin{align}
        \| \epsilon^{(R)} \|    &\leq \exp{-\frac{(2^q\Delta -K-1/2)^2}{4 \sigma^2}} \cr
                            &\leq \exp{-\frac{(2K-K-1/2)^2}{4 \sigma^2}} \cr 
                            &\leq \exp{-\frac{(K-1/2)^2}{4 \sigma^2}} \cr 
    \end{align}
    Finally, we will impose the bound:
    \begin{align}
        \exp{-2\pi^2\sigma^2} \geq \exp{-\frac{(K-1/2)^2}{4\sigma^2}},
    \end{align}
    which leads to the tighter constraint on $\sigma$,
    \begin{align}
        \sigma  \leq \frac{1}{2^{1/4}}\sqrt{\frac{K-1/2}{2\pi}}.
    \end{align}
With this, we have the relative standard deviation
    \begin{align}
        \tilde\sigma  \leq \frac{1}{2^{1/4}}\sqrt{\frac{\Delta/3-1/2^{q+1}}{2\pi (2^q)}}.
    \end{align}
If we assume that $1/2^q\leq \Delta / 3$, we can simplify the bound of the domain further with
    \begin{align}
        \tilde\sigma  \leq \frac{1}{2^{1/4}} \sqrt{\frac{1}{2^q}}\sqrt{\frac{\Delta}{12\pi}},
    \end{align}
at the cost of reducing the domain.
This way we can upper-bound both exponentials with a single one:
    \begin{align}
        \exp{-2 \pi^2 \sigma^2 }&=\exp{-2 \pi^2 \tilde\sigma^2 (2^q)^2 },\cr 
        \abs{ \tilde{G}_0 - \tilde{F}_0 }\leq\sqrt{\abs{ \tilde{G}_0 - \tilde{F}_0 }}&\leq\exp{-2 \pi^2 \tilde\sigma^2 (2^q)^2 } \cr 
        \|\epsilon^{(R)}\|&\leq\exp{-2 \pi^2 \tilde\sigma^2 (2^q)^2 }
    \end{align}

Morever, switching to the relative error $\tilde\varepsilon_m$ 
\begin{align}
\tilde\varepsilon_m & \leq (2^q)^{m} (m!)^2 \exp{-2 \pi^2 \tilde\sigma^2 (2^q)^2 } C(\eta),
\end{align}
where
\begin{align}
    C(\eta) &= \left[\left(\frac{128}{45}\right)\exp{12}  \left(12 + 3 + \frac{9}{4}\sqrt{\frac{1}{\eta}}   \right) \right.\cr
    &\left. +10\exp{12} + \frac{55}{8} \exp{2} \left(  1  + \sqrt{\frac{5}{3}}\sqrt{\frac{1}{\eta}} \right)\right].
\end{align}

Solving for $2^q$
\begin{align}
\left(\frac{\tilde\varepsilon_m}{(m!)^2 C(\eta)}\right)^{2/m} & \leq (2^q)^{2}  \exp{-\frac{4 \pi^2 \tilde\sigma^2 (2^q)^2}{m} }  \cr
-\frac{4 \pi^2 \tilde\sigma^2 }{m}\left(\frac{\tilde\varepsilon_m}{(m!)^2 C(\eta)}\right)^{2/m} & \geq -\frac{4 \pi^2 \tilde\sigma^2 (2^q)^2}{m}  \exp{-\frac{4 \pi^2 \tilde\sigma^2 (2^q)^2}{m} }  \cr
-\frac{4 \pi^2 \tilde\sigma^2 }{m}\left(\frac{\tilde\varepsilon_m}{(m!)^2 C(\eta)}\right)^{2/m} & \geq -\frac{4 \pi^2 \tilde\sigma^2 (2^q)^2}{m}  \exp{-\frac{4 \pi^2 \tilde\sigma^2 (2^q)^2}{m} }  \cr
W_{-1}\left(-\frac{4 \pi^2 \tilde\sigma^2 }{m}\left(\frac{\tilde\varepsilon_m}{(m!)^2 C(\eta)}\right)^{2/m}\right) & \leq -\frac{2 \pi^2 \tilde\sigma^2 (2^q)^2}{m}    \cr
\frac{4 \pi^2 \tilde\sigma^2 (2^q)^2}{m}& \leq -W_{-1}\left(-\frac{4 \pi^2 \tilde\sigma^2 }{m}\left(\frac{\tilde\varepsilon_m}{(m!)^2 C(\eta)}\right)^{2/m}\right)     \cr
 (2^q)^2& \leq \frac{m}{4\pi^2\tilde\sigma^2}\left(-W_{-1}\left(-\frac{4 \pi^2 \tilde\sigma^2 }{m}\left(\frac{\tilde\varepsilon_m}{(m!)^2 C(\eta)}\right)^{2/m}\right) \right)     \cr
\end{align}
We note that
$$
1 + \sqrt{2u} + \tfrac{2}{3}u < -W_{-1}\left(-e^{-u-1}\right) <  1 + \sqrt{2u} + u < 1 + 3u\quad \text{for } u > 0.
$$
In this scenario,
$$
-u-1 = \log{\frac{4 \pi^2 \tilde\sigma^2 }{m}\left(\frac{\tilde\varepsilon_m}{(m!)^2 C(\eta)}\right)^{2/m}}
$$
$$
u+1 = \log{\frac{m}{4\pi^2 \tilde\sigma^2 }\left(\frac{(m!)^2 C(\eta)}{\tilde\varepsilon_m}\right)^{2/m}}
$$
$$
u = \log{\frac{m}{4 e \pi^2 \tilde\sigma^2 }\left(\frac{(m!)^2 C(\eta)}{\tilde\varepsilon_m}\right)^{2/m}}.
$$
Thus,
\begin{align}
     2^q& \leq \frac{\sqrt{m}}{2\pi \tilde\sigma}\sqrt{1+3\log{\frac{m}{4 e \pi^2 \tilde\sigma^2 }\left(\frac{(m!)^2 C(\eta)}{\tilde\varepsilon_m}\right)^{2/m}}}\cr  
     &\quad\quad \text{for }\frac{m}{2 e \pi^2 \tilde\sigma^2 }\left(\frac{(m!)^2 C(\eta)}{\tilde\varepsilon_m}\right)^{2/m} > 1.    \cr
\end{align}

\end{proof}

Fix the integration interval to $2K+1 = \kappa \sigma = \lfloor (2/3) \Delta 2^q \rfloor$, such that we can add-in all the outcomes that lie $\lfloor (2/3) \Delta  2^q \rfloor-1$ to the right without fear of adding contributions from higher states. Moreover, there should be a $\lfloor \Delta 2^q/3 \rfloor$ segment that is dark (no counts) which should serve as a test of validity of the lowest energy outcome.

In order to have negligible excited state contamination $\|\epsilon^{(R)}\|$ and probability leakage $\sqrt{\abs{ \tilde{G}_0 - \tilde{F}_0 }}$, $1/\sigma_0$ should at least be $ O(1/\Delta)$.

\subsection{Estimates on the failure rate of \Cref{alg:sampl_n_trim}}

\begin{lemma}[Failure rate]\label{lem:fail_rate}

Provided that $\abs{G_0(0)-\Tilde{G}_0} \leq 1/8$, $\abs{ \tilde{G}_0 - \tilde{F}_0 } \leq 1/8$, $1/2^q \leq \Delta/6$, the number of samples is $M_0=\left\lceil\frac{16}{3 \eta}\log{\frac{3}{ \delta }}\right\rceil$, and that the relative width is 
$$
\tilde\sigma \leq \frac{ \Delta/6 }{\sqrt{2 \log{\frac{4 M_0}{\delta} \left( 1 + \sqrt{\frac{5}{3}}\sqrt{\frac{1-\eta}{\eta}}\right)^2} } },
$$
we will have that the probability of failing, that is, that we get an outcome from the left, from the gap window, or that we get zero outcomes from the left half of the window of interest $x^{\rm (left)}\in \{-K+x_0,-K+x_0+1,\dots,x_0\}$ is, by union bound,
$$
p_{\rm fail} \leq \delta
$$
\end{lemma}
\begin{proof}
We first upper-bound the probability of getting a count within the gap of length $K$ above from our $2K+1$-window of interest 
\begin{align}
p_{\rm gap\,\,window} &= \frac{\abs{\gamma_0}^2}{\mathcal{N}}\sum^{2K}_{K+1} \abs{f^{\rm polut}}^2 \cr
& \leq \frac{\abs{\gamma_0}^2}{\mathcal{N}}\sum^{2K}_{K+1} \abs{f}^2 + \frac{\mathcal{N}}{\abs{\gamma_0}^2}\abs{\epsilon^{\rm (polut)}}^2 + 2 \frac{\sqrt{\mathcal{N}}}{\abs{\gamma_0}}\abs{f}\abs{\epsilon^{\rm (polut)}} \cr
& \leq \frac{\abs{\gamma_0}^2}{\mathcal{N}}\left( \sum^{\infty}_{K+1} \abs{f}^2 + \frac{\mathcal{N}}{\abs{\gamma_0}^2}\sum^{2K}_{-\infty}\abs{\epsilon^{\rm (polut)}}^2 + 2 \frac{\sqrt{\mathcal{N}}}{\abs{\gamma_0}}\sum^{2K}_{K+1} \abs{f}\abs{\epsilon^{\rm (polut)}}\right) \cr
& \leq \frac{\abs{\gamma_0}^2}{\mathcal{N}}\left( \sum^{\infty}_{K+1} \abs{f}^2 + \frac{\mathcal{N}}{\abs{\gamma_0}^2}\sum^{2K}_{-\infty}\abs{\epsilon^{\rm (polut)}}^2 + 2 \frac{\sqrt{\mathcal{N}}}{\abs{\gamma_0}}\left(\sum^{2K}_{K+1} \abs{f}^2\right)^{1/2}\left(\sum^{2K}_{K+1} \abs{\epsilon^{\rm (polut)}}^2\right)^{1/2}\right) \cr
& \leq \frac{\abs{\gamma_0}^2}{\mathcal{N}}\left( \sum^{\infty}_{K+1} g_0(n,\tilde{\mu} ) + \frac{\mathcal{N}}{\abs{\gamma_0}^2}\frac{1-\eta}{\inf_{\tilde\mu}(\mathcal{N})}\sum^{2K}_{-\infty}g_0(n,\tilde{\mu}+2^q\Delta ) \right.\cr  
&\left.+ 2 \frac{\sqrt{\mathcal{N}}}{\abs{\gamma_0}}\sqrt{\frac{1-\eta}{\inf_{\tilde\mu}(\mathcal{N})}}\left(\sum^{\infty}_{K+1} g_0(n,\tilde{\mu} )\right)^{1/2}\left(\sum^{2K}_{-\infty} g_0(n,\tilde{\mu}+2^q\Delta )\right)^{1/2}\right) \cr
\end{align}

\begin{align}\label{eq:right_sum}
    \sum^{\infty}_{K+1} g_0(x,\tilde{\mu} )  &\leq
    \frac{1}{\sigma\sqrt{2\pi}}\int^{\infty}_{K} \exp{-\frac{(x-\Tilde{\mu})^2}{2\sigma^2}} {\rm d} x \cr  
    &\leq \erfc\left(\frac{K-1/2}{\sqrt{2}\sigma}\right) \leq \exp{-\frac{(K-1/2)^2}{2\sigma^2}}
\end{align}

\begin{align}\label{eq:left_sum}
    \sum^{2K}_{-\infty} g_0(x,\tilde{\mu} +2^q\Delta)  &\leq
    \frac{1}{\sigma\sqrt{2\pi}}\int^{2K+1}_{-\infty} \exp{-\frac{(x-\Tilde{\mu}-2^q\Delta)^2}{2\sigma^2}}{\rm d} x \cr  
    &\leq \erfc\left(\frac{2^q\Delta -2K-1/2}{\sqrt{2} \sigma}\right) \leq \exp{-\frac{(2^q\Delta -2K-1/2)^2}{2 \sigma^2}}.
\end{align}
Remembering that,
\begin{align}\label{eq:K_Delta}
2K + 1 &= \lfloor (2/3) 2^q \Delta \rfloor,
\end{align}
which in turn means
\begin{align}
2K + 1 &\leq  (2/3) 2^q \Delta \cr
K &\leq  (1/3) 2^q\Delta - 1/2 
\end{align}
we then use it on \Cref{eq:left_sum} to get
\begin{align}
    \sum^{2K}_{-\infty} g_0(x,\tilde{\mu} +2^q\Delta)  & \leq \exp{-\frac{\left(2^q\Delta/3 + 1/2\right)^2}{2 \sigma^2}}.
\end{align}
\Cref{eq:K_Delta} also means that
\begin{align} 
2K + 1 &\geq  (2/3) 2^q \Delta - 1 \cr
K &\geq  (1/3) 2^q\Delta-1,
\end{align}
thus, if we use this in \Cref{eq:right_sum} we obtain
\begin{align}
    \sum^{\infty}_{K+1} g_0(x,\tilde{\mu} )  &\leq \exp{-\frac{\left(2^q \Delta/3 -3/2\right)^2}{2\sigma^2}}.
\end{align}
Finally, for $p_{\rm gap\,\,window}$ we obtain the bound
\begin{align}
p_{\rm gap\,\,window} &  \leq \frac{\abs{\gamma_0}^2}{\mathcal{N}}\exp{-\frac{\left(2^q \Delta/3 -3/2\right)^2}{2\sigma^2}}\cr 
&\left( 1 + \frac{\mathcal{N}}{\abs{\gamma_0}^2}\frac{1-\eta}{\inf_{\tilde\mu}(\mathcal{N})}   
+ 2 \frac{\sqrt{\mathcal{N}}}{\abs{\gamma_0}}\sqrt{\frac{1-\eta}{\inf_{\tilde\mu}(\mathcal{N})}}\right)  \cr
\end{align}
If we use \Cref{lem:norm_error} while assuming $\abs{G_0(0)-\Tilde{G}_0} \leq 1/8$ and $\abs{ \tilde{G}_0 - \tilde{F}_0 } \leq 1/8$, and also use $\abs{\gamma_0}^2\leq 1$, and $0<\eta \leq \abs{\gamma_0}^2$, we obtain:
\begin{align}
p_{\rm gap\,\,window} &  \leq \frac{4}{3}\exp{-\frac{\left(2^q \Delta/3 -3/2\right)^2}{2\sigma^2}}\cr 
&\left( 1 + \sqrt{\frac{5}{3}}\sqrt{\frac{1-\eta}{\eta}}\right)^2  \cr
\end{align}
Similarly, the probability of obtaining an outcome from the left to our window of interest is 
\begin{align}
p_{\rm left} &\leq \frac{1}{\min_{\tilde\mu}\mathcal{N}}\sum^{-K-1}_{-\infty} \abs{f}^2 \leq \frac{1}{\min_{\tilde\mu}\mathcal{N}}\erfc\left(\frac{K-1/2}{\sqrt{2}\sigma}\right) \cr 
&\leq\frac{1}{\min_{\tilde\mu}\mathcal{N}}\exp{-\frac{(K-1/2)^2}{2\sigma^2}} \leq \frac{4}{3}\exp{-\frac{\left(2^q \Delta/3 -3/2\right)^2}{2\sigma^2}}.
\end{align}
To simplify calculations we will loosen the bound on $p_{\rm left}$ to match that on $p_{\rm gap}$.
Now, we take a look at the hit-rate from \Cref{lem:hit_rate} and assert that the probability of getting an outcome from $x^{\rm(left)} = \{-K+x_0,-K+x_0+1,\dots,x_0\}$, the left side of the window of intrest, is
\begin{align}
    p_{x^{\rm (left)}} \geq \tilde{p}_0/2
\end{align}
If we perform the Gaussian phase estimation experiment $M_0$ times, the probability of getting at least one outcome from $x^{\rm(left)}$ is
\begin{align}
   P(X\geq 1) = 1 - P(X=0) = 1 - (1- p_{x^{\rm (left)}})^{M_0}.
\end{align}
The corresponding error rate is
\begin{align}
    p_{\rm zero} &= P(X=0) = (1- p_{x^{\rm (left)}})^{M_0} \cr
                &\leq \exp{-p_{x^{\rm (left)}} M_0}\cr
                &\leq \exp{-p_0 M_0 /2}
\end{align}
Now, from \Cref{cor:hit_rate}, we have that
\begin{align}
    p_{\rm zero} &\leq \exp{-3 \eta M_0 /16}    
\end{align}
If upper bound $p_{\rm zero}$ through the quantity $\delta$ the following way
\begin{align}
    p_{\rm zero} \leq & \exp{-3 \eta M_0 /16} \leq \frac{1}{3} \delta = \exp{-3 \eta \tilde{M}_0 /16}
\end{align}
where $M_0=\lceil \tilde{M}_0 \rceil$. We solve for the number of samples $M_0$ obtaining
\begin{align}
    M_0 = \left\lceil\frac{16}{3 \eta}\log{\frac{3}{ \delta }}\right\rceil.
\end{align}
Now, we relate $\delta$ to $p_{\rm gap}$ with
\begin{align}
p_{\rm gap} \leq \frac{4}{3}\exp{-\frac{\left(2^q \Delta/3 -3/2\right)^2}{2\sigma^2}} \left( 1 + \sqrt{\frac{5}{3}}\sqrt{\frac{1-\eta}{\eta}}\right)^2 \leq \frac{1}{3} (\delta /M_0 )  \cr
\end{align}

Solving for $\tilde\sigma=\sigma/2^q$
\begin{align}
    \tilde\sigma \leq \frac{\left( \Delta/3 -3/2^{q+1}\right)}{\sqrt{2 \log{\frac{4 M_0}{\delta} \left( 1 + \sqrt{\frac{5}{3}}\sqrt{\frac{1-\eta}{\eta}}\right)^2} } }
\end{align}
We impose the constraint
$$
1/2^q \leq \Delta/6
$$
Thus, obtaining the decoupled bound
\begin{align}
    \tilde\sigma \leq \frac{ \Delta/6 }{\sqrt{2 \log{\frac{4 M_0}{\delta} \left( 1 + \sqrt{\frac{5}{3}}\sqrt{\frac{1-\eta}{\eta}}\right)^2} } }.
\end{align}
With this, we now have, the event of failure is qualified by
\begin{align}
    A_{\rm fail} = A_{\rm zero} \cup \left(\bigcup_i A_{i,\rm left}\right) \cup \left(\bigcup_i A_{i,\rm gap}\right)
\end{align}
Here, $A_{zero}$ is the event that no outcomes from the window $x^{\rm(left)}=\{-K+x_0,-K+x_0+1,\dots,x_0\}$ after $M_0$ measurements, $A_{i,\rm left}$ and $A_{i,\rm gap}$ are the events in which we get an outcome from the left and the gap window respectively on $i$th measurement. By the union bound we get,
\begin{align}
    p_{\rm fail} \leq p_{\rm zero} + M_0 p_{\rm left} + M_0 p_{\rm gap} \leq \delta.
\end{align}

\end{proof}
Lastly, after putting everything together, we get the formal version of \Cref{thm:short_sampling_algorithm}

\begin{theorem}
    \label{thm:final_sampling_algorithm}
There exists an algorithm to sample the $m_{\rm th}$ moment, where $m\geq 1$, of a random variable $x$ which approximately has the normal distribution $\mathcal{N}(\tilde\sigma,\mu=E_0)$, but which is bounded through $|x-\mu|\leq    (2/3)\Delta$. Here, $E_0$ is the ground state energy of $H$, and $\|H\|\leq 1-\Delta$. $\Delta$ is the lower bound on the spectral gap and $\eta\leq\abs{\gamma_0}^2$ is the lower bound on the initial state squared overlap with the ground state. $\delta\leq 0.01$ 
The number of samples required of the circuit \Cref{fig:gaussian_qpea} is
$$
    M_0=\left\lceil\frac{16}{3 \eta}\log{\frac{3}{ \delta }}\right\rceil
$$
where $\delta$ is the target overall failure rate of the algorithm, in order to obtain at least one approximate sample of the $m_{\rm th}$ moment whose expectation value has the error:
\begin{align}
   \tilde\varepsilon_m = \abs{\mu_m-\mathbb{E}(x^m)}.
\end{align}
By fixing the standard deviation to
\begin{align*}
    \tilde\sigma = \frac{1}{\sqrt{m}}\frac{ \Delta/6 }{\sqrt{2 \log{\frac{4 M_0}{\delta} \left( 1 + \sqrt{\frac{5}{3}}\sqrt{\frac{1-\eta}{\eta}}\right)^2} } },
\end{align*}
we require an ancillary register size of
    \begin{align*}
         q &= \left\lceil{\rm log}_2\left(\frac{3 m}{\pi \Delta}\sqrt{1+3\log{\frac{m}{4 e \pi^2 \tilde\sigma^2 }\left(\frac{(m!)^2 C(\eta)}{\tilde\varepsilon_m}\right)^{2/m}}} \sqrt{2 \log{\frac{4 M_0}{\delta} \left( 1 + \sqrt{\frac{5}{3}}\sqrt{\frac{1-\eta}{\eta}}\right)^2} } \right) \right\rceil \cr  
         &\qquad\qquad \textup{for}\,\,\frac{m}{2 e \pi^2 \tilde\sigma^2 }\left(\frac{(m!)^2 C(\eta)}{\tilde\varepsilon_m}\right)^{2/m} > e, 
\end{align*}
where
\begin{align*}
    C(\eta) &= \left[\left(\frac{128}{45}\right)\exp{12}  \left(12 + 3 + \frac{9}{4}\sqrt{\frac{1}{\eta}}   \right) \right.\cr
    &\left. +10\exp{12} + \frac{55}{8} \exp{2} \left(  1  + \sqrt{\frac{5}{3}}\sqrt{\frac{1}{\eta}} \right)\right],
\end{align*}
while also maintaining, $\abs{G_0(0)-\Tilde{G}_0} \leq 1/8$, $\abs{ \tilde{G}_0 - \tilde{F}_0 } \leq 1/8$, $\| \epsilon^{(R)} \|/\sqrt{\eta} \leq 1/8$, and
\begin{align*}
    \frac{\sqrt{ \log{\frac{4 M_0}{\delta} \left( 1 + \sqrt{\frac{5}{3}}\sqrt{\frac{1-\eta}{\eta}}\right)^2} }}{\sqrt{1+3\log{\frac{m}{4 e \pi^2 \tilde\sigma^2 }\left(\frac{(m!)^2 C(\eta)}{\tilde\varepsilon_m}\right)^{2/m}}}} \geq 2.
\end{align*}
\end{theorem}

\begin{proof}
We now group all the requirements by previous lemmas. The number of samples required of the circuit \Cref{fig:gaussian_qpea} is (See \Cref{lem:fail_rate})
$$
    M_0=\left\lceil\frac{16}{3 \eta}\log{\frac{3}{ \delta }}\right\rceil
$$
where $\delta$ is the target overall failure rate of the algorithm, in order to obtain at least one approximate sample of the $m_{\rm th}$ moment whose expectation value has the error:
\begin{align}
   \tilde\varepsilon_m = \abs{\mu_m-\mathbb{E}(x^m)}.
\end{align}
The requirements of the method are (See \Cref{lem:fail_rate})
\begin{align}\label{eq:tsigma_bound}
    \tilde\sigma \leq \frac{ \Delta/6 }{\sqrt{2 \log{\frac{4 M_0}{\delta} \left( 1 + \sqrt{\frac{5}{3}}\sqrt{\frac{1-\eta}{\eta}}\right)^2} } },
\end{align}
and a circuit depth/ancillar register size of (See \Cref{thm:2q})
    \begin{align}\label{eq:lowerbound_on_q}
         2^q& \geq \frac{\sqrt{m}}{2\pi \tilde\sigma}\sqrt{1+3\log{\frac{m}{4 e \pi^2 \tilde\sigma^2 }\left(\frac{(m!)^2 C(\eta)}{\tilde\varepsilon_m}\right)^{2/m}}}\cr  
         &\quad\quad \textup{for}\,\,\frac{m}{2 e \pi^2 \tilde\sigma^2 }\left(\frac{(m!)^2 C(\eta)}{\tilde\varepsilon_m}\right)^{2/m} > 1,    \cr
    \end{align}
    
while also maintaining $1/2^q \leq \Delta/6$ (See \Cref{lem:fail_rate}), $\abs{G_0(0)-\Tilde{G}_0} \leq 1/8$ (See \Cref{lem:fail_rate}), $\abs{ \tilde{G}_0 - \tilde{F}_0 } \leq 1/8$ (See \Cref{lem:fail_rate}), $\| \epsilon^{(R)} \|/\sqrt{\eta} \leq 1/8$ (See \Cref{thm:2q}), and (See \Cref{thm:2q})
\begin{align}\label{eq:upperbound_on_q}
        \tilde\sigma  \leq \frac{1}{2^{1/4}}\sqrt{\frac{1}{2^q}}\sqrt{\frac{\Delta}{12\pi}}.
\end{align}

We now enforce \Cref{eq:tsigma_bound} by setting
\begin{align}
    \tilde\sigma = \frac{1}{\sqrt{m}}\frac{ \Delta/6 }{\sqrt{2 \log{\frac{4 M_0}{\delta} \left( 1 + \sqrt{\frac{5}{3}}\sqrt{\frac{1-\eta}{\eta}}\right)^2} } }.
\end{align}

We now replace the definition of $\tilde\sigma$ in \Cref{eq:upperbound_on_q} to obtain
    \begin{align}\label{eq:upperbound_on_q_v2}
         2^q \leq \frac{3m}{\pi\Delta}\sqrt{ \log{\frac{4 M_0}{\delta} \left( 1 + \sqrt{\frac{5}{3}}\sqrt{\frac{1-\eta}{\eta}}\right)^2} }\sqrt{2 \log{\frac{4 M_0}{\delta} \left( 1 + \sqrt{\frac{5}{3}}\sqrt{\frac{1-\eta}{\eta}}\right)^2} }.
\end{align}

Now, we must fulfill both \Cref{eq:upperbound_on_q_v2} and \Cref{eq:lowerbound_on_q}, which we accomplish through setting
  \begin{align}\label{eq:q_definition}
         q& = \left\lceil {\rm log}_2\left(\frac{3 m}{\pi \Delta}\sqrt{1+3\log{\frac{m}{4 e \pi^2 \tilde\sigma^2 }\left(\frac{(m!)^2 C(\eta)}{\tilde\varepsilon_m}\right)^{2/m}}} \sqrt{2 \log{\frac{4 M_0}{\delta} \left( 1 + \sqrt{\frac{5}{3}}\sqrt{\frac{1-\eta}{\eta}}\right)^2} }\right)\right\rceil
\end{align}
and enforcing:
\begin{align}
    \frac{\sqrt{ \log{\frac{4 M_0}{\delta} \left( 1 + \sqrt{\frac{5}{3}}\sqrt{\frac{1-\eta}{\eta}}\right)^2} }}{\sqrt{1+3\log{\frac{m}{4 e \pi^2 \tilde\sigma^2 }\left(\frac{(m!)^2 C(\eta)}{\tilde\varepsilon_m}\right)^{2/m}}}} \geq 2.
\end{align}
We also enforce $1/2^q \leq \Delta/6$  by requiring $\delta \leq 0.01$ and $\frac{m}{2 e \pi^2 \tilde\sigma^2 }\left(\frac{(m!)^2 C(\eta)}{\tilde\varepsilon_m}\right)^{2/m} > e$.
\end{proof}

\end{document}